\documentclass[cmp]{svjour}

\usepackage{times,amssymb,latexsym,amsmath,graphics,setspace,subfigure}
\usepackage{epsfig,epsf,amsfonts}

\journalname{Communications in Mathematical Physics}

\begin{document}

\title{Ergodicity of some open systems with particle-disk interactions.}
\author{Tatiana Yarmola}
\institute{Department of Mathematics, Mathematics Building, University of Maryland, College Park, MD 20742, USA \\
\email{yarmola@math.umd.edu}}

%\date{}
%\communicated{}

\maketitle
\begin{abstract}
We consider steady states for a class of mechanical systems with particle-disk interactions coupled to two, possibly unequal, heat baths. We show that any steady state that satisfies some natural assumptions is ergodic and absolutely continuous with respect to a Lebesgue-type reference measure and conclude that there exists at most one absolutely continuous steady state.
\end{abstract}

\section*{Introduction.}
Explaining macroscopic phenomena of open systems from the microscopic dynamics is an intriguing subject in statistical mechanics. By an open system we mean a deterministic and energy conserving system that exchanges energy and matter with multiple heat baths. Some examples of such systems have been studied \cite{Balint,Collet,Anharmonic1,Eckmann,Anharmonic2,Anharmonic3,Nonequilibrium,Klages_Nicolis_Rateitschak,Larralde_Leyvraz_MejiaMonasterio,Models,Anharmonic}, but a good understanding of non-equilibrium behavior, i.e. when the heat bath are \emph{not} equal, has not been reached yet.

In this paper we consider a mechanical system that attempts to reproduce the phenomenological laws of thermodynamic transport \cite{Larralde_Leyvraz_MejiaMonasterio}. It consists of a rectangular domain containing a chain of $N$ identical pinned down disk scatterers that are allowed to rotate freely. See Figure \ref{fig:rectangle}. The particles in the system bounce elastically from the walls and exchange energy with the disks through \textquotedblleft perfectly rough" collisions \cite{Klages_Nicolis_Rateitschak,Larralde_Leyvraz_MejiaMonasterio}. The system is coupled to two possibly unequal heat baths through the openings, which correspond to the "left" and the "right" sides of the rectangle. Once a particle reaches an opening, it leaves the system forever. New particles can be introduced to the system by the heat baths; they are emitted at random times according to some probability distributions that describe injection positions and velocities. Neither energy nor the number of particles is conserved, so the system can be loosely referred to as a \textquotedblleft grand-canonical ensemble". The detailed settings are available in section \ref{sect: settings}.

This paper demonstrates rigorous results regarding invariant densities and ergodicity of the steady states, which we refer to as invariant measures.  Theorem \ref{thm:main} provides conditions under which an invariant measure is guaranteed to be ergodic and absolutely continuous with respect to a natural reference measure. From Theorem \ref{thm:main} we immediately obtain uniqueness of the absolutely continuous measure as well as the ergodic decomposition. In the equilibrium situation, an absolutely continuous invariant measure can be written down explicitly \cite{Nonequilibrium}. By Theorem \ref{thm:main} we conclude that this invariant measure is ergodic.

The question of existence and uniqueness of the steady states for other types of open systems that attempt to model thermodynamic transport has been studied to some extent. The best studied examples are chains of anharmonic oscillators coupled
at both ends to heat reservoirs \cite{Anharmonic}. The existence and uniqueness results for such systems have been obtained in \cite{Anharmonic1,Anharmonic2,Anharmonic3} under rather restrictive assumptions on the potential in the chain and coupling to the reservoirs.

For open particle systems with indirect particle interactions, existence of the stationary states has only been demonstrated in equilibrium situations by providing them explicitly \cite{Balint,Nonequilibrium}. For non-equilibrium states, existence turns out to be a nontrivial mathematical problem which requires a detailed understanding of the dynamics and we leave it for future work.

Uniqueness results for open particle systems are scarce. Ergodicity and thus uniqueness of natural invariant measures for a $1$-dimensional particle model was obtained in \cite{Balint}; we borrow some of the ideas developed there in our proof. For planar geometries, an important first step towards showing ergodicity we rely on was done in \cite{Eckmann}, where the authors demonstrated that the action of the baths can drive the system from \textquotedblleft almost" any state to \textquotedblleft almost" any other state in a finite time. We chose simpler geometry in our model in order to focus on the essential properties of the system that lead to absolute continuity and ergodicity of the invariant measures leaving geometric complications aside. However, combined with the results in \cite{Eckmann}, only minor modifications are needed for our argument to carry through for the class of systems described there. Our argument might potentially be generalized to a wider class of geometries, e.g. presented in \cite{Nonequilibrium,Models}.

The organization of our paper is the following. Section \ref{sect: settings} contains a detailed description of the model. We state the main results in section \ref{sect: results}. The remaining part of the paper is devoted to the proof of Theorem \ref{thm:main}. We start with an outline of proof in section \ref{sect:outline of proof}, where we state three propositions and show how they imply Theorem \ref{thm:main}. The propositions are proven in sections \ref{sect:abs cont part}, \ref{sect: empty playground acquiring density}, \ref{sect: flush particles out}, and \ref{sect: proof flush particles out}.
\begin{center}
\begin{figure}
  \begin{center}
  \includegraphics[width=4in]{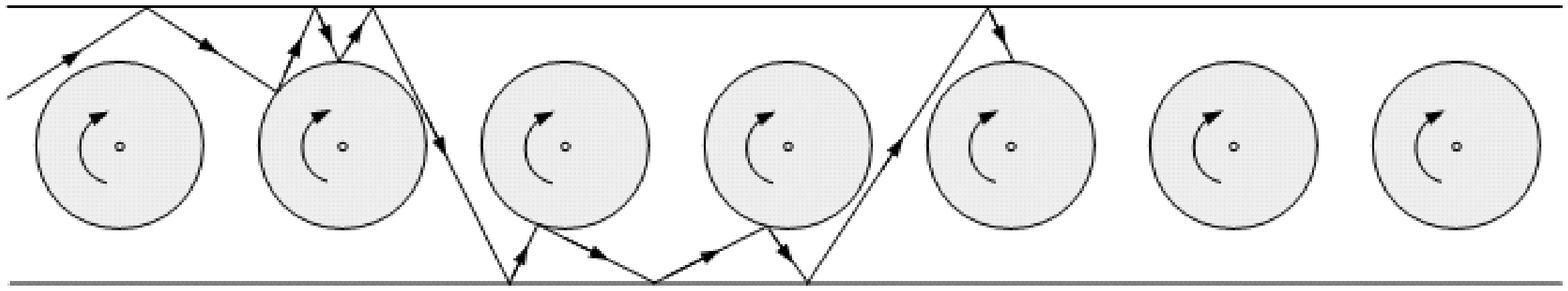} \\
  \end{center}
  \caption{1D array of disks in a rectangle}\label{fig:rectangle}
\end{figure}
\end{center}

\section{Model Description.} \label{sect: settings}
\subsection{Dynamics of closed systems} \label{subsect: closed playground}

Let $\Gamma_0$ be a rectangle bounded by $y=\pm 1$, $x=0$ and $x=2N$, where $N$ is an arbitrary positive integer. In the interior of $\Gamma_0$ lie $N$ disks $D_j$, $1 \leq j \leq N$, of equal radii $R<1$ centered at $(2j-1,0) \in \Gamma_0$, $1 \leq j \leq N$. The centers of the disks are fixed and the disks are allowed to rotate freely around the center, each carrying a finite amount of kinetic energy derived from its angular velocity. Denote the states of the disks by $(\varphi_j,\omega_j)$, $1 \leq j \leq N$, where $\varphi_j$ is the disk's angular position relative to a marked reference point and $\omega_j$ is the disk's angular velocity.

A number of particles move around in the playground $\Gamma = \overline{\Gamma_0 \setminus \cup_{j=1}^N D_j}$, with particle positions $q_i \in \Gamma$ and velocities $v_i \in \mathbb{R}^2$. Apart from collisions with the boundary of the playground $\partial \Gamma = \partial \Gamma_0 \cup (\cup_{j=1}^N \partial D_j)$ the particles are assumed to move freely with constant velocities; particles do not interact with each other. The collisions with $\partial \Gamma_0$ are specular, and upon collision of a particle with a disk, a certain energy exchange occurs \cite{Klages_Nicolis_Rateitschak,Larralde_Leyvraz_MejiaMonasterio}. More precisely:

The phase space of such system with $k$ particles is
$$\tilde{\Omega}_k=(\Gamma^k \times \partial D_1 \times \cdots \times \partial D_N \times \mathbb{R}^{2k+N})/\sim,$$ where \\
$\texttt{q} = (q_1, \cdots, q_k) \in \Gamma^k$ denotes the positions of $k$ particles, \\
$\varphi = (\varphi_1, \cdots, \varphi_N) \in \partial D_1 \times \cdots \times \partial D_N$ denotes the angular positions of the disks, \\
$\texttt{v}=(v_1, \cdots, v_k) \in \mathbb{R}^{2k}$ denotes the velocities of $k$ particles, \\
$\omega = (\omega_1, \cdots, \omega_N) \in \mathbb{R}^N$ denotes the angular velocities of the $N$ disks, \\
and $\sim$ is the relation identifying pairs of points on the collision manifold: $M_k=\{(\texttt{q},\varphi,\texttt{v},\omega): q_i \in \partial \Gamma$ for some $i\}$ with the rules of identification as follows:

Let $v_i=((v_i)_t,(v_i)_\perp)$ be the tangential and the normal components of $v_i$. If a particle collides with the boundary of the playground, $q_i \in \Gamma_0$, then the angle of reflection is equal to the angle of incidence, i.e.
$$(v_i)_\perp '=-(v_i)_\perp, \; \; \; (v_i)_t'=(v_i)_t$$
If collision with a disk occurs, $q_i \in \partial D_j$ for some $j$, then
$$(v_i)_\perp'=-(v_i)_\perp, \; \; \; (v_i)_t'=(v_i)_t - \frac{2\eta}{1+\eta}((v_i)_t-R \omega_j),$$
$$R \omega_j'=R \omega_j + \frac{2}{1+\eta}((v_i)_t-R \omega_j),$$
where $\eta=\frac{\Theta}{m R^2}$ is a dimensionless parameter relating
the moment of inertia of the disc $\Theta$, the mass of the particle $m$, and the radius of the disc $R$ \cite{Klages_Nicolis_Rateitschak,Larralde_Leyvraz_MejiaMonasterio}. Throughout the paper we assume that $0< \eta< \infty$.

The identification $\sim$ in the definition of $\tilde{\Omega}_k$ is as follows: $X \sim X'$, $X,X' \in M_k$, if the coordinates of $X$ and $X'$ are equal except for $v_i$'s such that $q_i \in \partial \Gamma$. If $q_i \in \Gamma_0$, we replace $v_i$ in $X$ by $v_i'$ in $X'$. If $q_i \in \partial D_j$, we replace $v_i$ and $\omega_j$ in $X$ by $v_i'$ and $\omega_j'$ in $X'$. Note that simultaneous collisions of several particles with the same disk are not defined, while there is no problem with simultaneous collisions with different disks and/or $\partial \Gamma_0$.

Define the discontinuous flow $\tilde{\Phi}_\tau$ on $\tilde{\Omega}_k$ by
$$\tilde{\Phi}_\tau(\texttt{q},\varphi,\texttt{v},\omega)=(\texttt{q}+\texttt{v}\tau,\varphi+\omega\tau, \texttt{v},\omega)$$
if no collisions are involved. When collisions occur, the rules of identification are given above. Then $m_k$ is an invariant measure for $\tilde{\Phi}_\tau$ on $\tilde{\Omega}_k$, where
$$\tilde{m}_k=(\lambda_2|_{\Gamma})^k \times \rho|_{\partial D_1} \times \cdots \times \rho|_{\partial D_N}  \times \lambda_{2k+N}.$$
Here $\lambda_d$ is $d$-dimensional Lebesgue measure and $\rho|_{\partial D_j}$ is the uniform measure on the circle $\partial D_j$ \cite{Nonequilibrium}.

\subsection{Dynamics of open systems} \label{subsect: coupling}
\subsubsection{Coupling to heat baths}

Suppose the rectangle $\Gamma_0$ has two openings, $\gamma_L=\{0\} \times [-1,1] \in \partial \Gamma_0$ and $\gamma_R=\{2N\} \times [-1,1] \in \partial \Gamma_0$, each connected to a heat bath that absorbs and emits particles. A particle absorbed by one of the baths leaves the system forever. The injection process for each bath is characterized by the following parameters: \\ \\
$\varrho_L$ and $\varrho_R$ - the injection rates of the baths. The injection processes are Poisson with rates $\varrho_L$ and $\varrho_R$ respectively. \\
$\Upsilon_L(\upsilon)$ and $\Upsilon_R(\upsilon)$ - the distributions of the positions of injection with values in $\gamma_L$ and $\gamma_R$. \\
$\Delta_L(\zeta)$ and $\Delta_L(\zeta)$ - the distributions of the angles of injection with values in $(-\frac{\pi}{2},\frac{\pi}{2})$. \\
$S_L(\varsigma)$ and $S_R(\varsigma)$ - the distributions of the injected particle speeds with values in $(0,\infty)$. \\

i.e. a particle is injected from the left bath at a random time $\tau \in (0,\infty)$ given by exponential distribution of rate $\varrho_L$, with random position $\xi \in \gamma_L$ drawn from $\Upsilon_L(\upsilon)$, at random angle $\delta$ drawn from $\Delta_L(\zeta)$ and with random speed $s$ drawn from $S_L(\varsigma)$. Similarly for the right bath.

We assume that each of the distributions $\Upsilon_L(\upsilon)$, $\Upsilon_R(\upsilon)$, $\Delta_L(\zeta)$, $\Delta_R(\zeta)$, $S_L(\varsigma)$, and $S_R(\varsigma)$ has positive density on the specified domains.

\subsubsection{The phase space} Since we are interested in the invariance properties, we would like to treat all particles in the system as identical and indistinguishable.
Let $\tilde{\Omega}_k$ be as in subsection \ref{subsect: closed playground}. Then the phase space of the system coupled to heat bath(s) is a disjoint union
$$\Omega = \sqcup_{k=0}^\infty \Omega_k,$$
where $\Omega_k$ is the quotient of $\tilde{\Omega}_k$ obtained by identifying the permutations of $k$ particles. If we denote unordered sets by $\{...\}$, under phase space, points in $\Omega$ are denoted by
$$X=(\{q_1, \cdots, q_k\},(\varphi_1, \cdots, \varphi_N), \{v_1, \cdots, v_k\}, (\omega_1, \cdots, \omega_N)),$$
with $v_j$ understood to be attached to $q_j$. Denote the quotient of the measure $\tilde{m}_k$ by $m_k$.

Let $\Phi_\tau$ be a continuous-time Markov process on $\Omega$ defined as follows:
\begin{itemize}
\item $\Phi_\tau$ is the quotient of $\tilde{\Phi}_\tau$ identifying the permutations while no particles enter or exit the system;
\item if $\Phi_\tau \in \Omega_k$ for all $0 \leq \tau \leq \tau_0$, and a particle exits the system at time $\tau_0$, then $\Phi_{\tau_0}(z)$ jumps to $\Omega_{k-1}$;
\item if $\Phi_\tau \in \Omega_k$ for all $0 \leq \tau \leq \tau_0$, and a particle is injected from one of the baths or point sources at time $\tau_0$, then $\Phi_{\tau_0}$ jumps to $\Omega_{k+1}$.
\end{itemize}

It is hard to write down the transition probabilities of $\Phi_\tau$ explicitly because particles may enter at any time. Denote the time-$\tau$ transition probability starting at state $X$ for $\Phi_\tau$ by $P^\tau_X$, if defined. $P^\tau_X$ if not defined for all $\tau$ and $X$: if, for example, $X$ is a state such that two particles will have their first collision at the same time $t$ with the same disk $D_j$, then $P^\tau_X$ is not defined for all $\tau \geq t$. Note that $P^\tau_X$ \emph{is} defined if the dynamics is defined with probability 1.

\begin{definition}
We will say that a measure $\mu$ is absolutely continuous with respect a measure $\nu$ if for any measurable set $A$, $\nu(A)=0 \; \Rightarrow \mu(A)=0$. We will say that two measures $\mu$ and $\nu$ defined on a set $C$ are singular if there exist sets $A$ and $B$, $A \cap B = \emptyset$, $A \cup B = C$ such that $\nu(A)=0$ and $\mu(B)=0$.
\end{definition}

When a measure $\mu$ is absolutely continuous with respect to $\nu$, we will denote it by $\mu \ll \nu$; we will also say that $\mu$ has a density with respect to $\nu$. If $\mu$ and $\nu$ are singular, we will denote it $\mu \perp \nu$.

\begin{definition} \label{defn: m}
Let $m$ be a measure on $\Omega$ such that for each nonnegative integer $k$, the conditional measure on $\Omega_k$ is $m_k$.
\end{definition}

For our purposes, $m$ is a natural reference measure: if $\mu$ is any measure on $\Omega$, we are interested whether $\mu$ is absolutely continuous with respect to $m$. In the rest of this paper, we will refer to $m$ as \emph{Lebesgue measure} and when we will say that $\mu$ is absolutely continuous without mentioning any reference measure, we would mean that $\mu$ is absolutely continuous with respect to $m$.

\subsection{Problems of interest}
Given a measure $\nu$, the push forward of $\nu$ under $\Phi_\tau$, if defined, is given by $(\Phi_\tau)_* \nu = \int_\Omega (P^\tau_X)(A)d\nu(X)$. Note that $P^\tau_X = (\Phi_\tau)_*\delta_X$ when defined.

\begin{definition}
A Borel probability measure $\mu$ on $\Omega$ is called \emph{invariant} under $\Phi_\tau$ if its push forward under $\Phi_\tau$ is defined for all $\tau$ and
$$\mu(A)=((\Phi_\tau)_* \mu)(A)=\int_\Omega (P^\tau_X)(A)d\mu(X).$$
\end{definition}

A classical way of analyzing flow $\Phi_\tau$ is studying the properties of its invariant measures such as existence, uniqueness, absolute continuity with respect to the Lebesgue measure $m$, and ergodicity. Proving the existence of invariant measures would require very technical arguments that involve dealing with tightness and discontinuities; we leave these problems for future work. In this paper we are going to focus on the issues of absolute continuity with respect to $m$ and ergodicity of the invariant measures provided they exist.

The main result of this paper claims that, for the class of systems in consideration, \emph{if} there exists an invariant measure $\mu$ such that the measure of the set of all states with \textquotedblleft trapped" particles is zero, then $\mu$ is both absolutely continuous and ergodic.

\section{Results.} \label{sect: results}

\begin{definition}
A state $X \in \Omega$ is said to contain a \emph{trapped particle} if either
\begin{itemize}
\item the velocity of the particle is zero, $v=0$, or
\item the $x$-component of the velocity of the particle is zero, $v_x=0$, and the position of the particle has its $x$-coordinate \textquotedblleft between the disks", i.e.
$$q \in ([0,1-R] \cup (\cup_{j=1}^N [2j-(1-R), 2j+(1-R)]) \cup [2N-(1-R),2N]) \times [-1,1].$$
\end{itemize}
Note that if we evolve the system starting from an initial state containing a trapped particle along any sample path, the system is going to contain a trapped particle at all times.
\end{definition}

Let $S_T \subset \Omega$ be the set of all states with trapped particles.

\begin{theorem} \label{thm:main}
Suppose there exists a probability measure $\mu$ invariant under $\Phi_\tau$ with $\mu(S_T)=0$. Then $\mu$ is absolutely continuous with respect to the Lebesgue measure $m$ and ergodic.
\end{theorem}

\begin{remark}
Since obviously $\mu(\Omega_k)\ne 0$ for all $k$, our assertion that $\mu$ is absolutely continuous with respect to $m$ means that on each $\Omega_k$, it has a density with respect to $m_k$.
\end{remark}

\begin{corollary} \label{corollary:main}
If $\nu$ is an ergodic measure for $\Phi_\tau$, then for some $k \geq 0$, $\nu$ is supported on $\cup_{j=k}^{\infty}\Omega_j$ and can be represented as a direct product of two measures, $\nu = \mu \times \pi$, such that $\mu$ is the unique absolutely continuous ergodic measure with $\mu(S_T)=0$ as in Theorem \ref{thm:main} and $\pi$ is a singular measure supported on the states in $\Omega_k$ traced by the trajectories of $k$ trapped particles.
\end{corollary}

\textbf{Equilibrium Case:}

Suppose the system is coupled to two \emph{equal} heat baths, characterized by temperature $T$ and injection rate $\varrho$, i.e. the injection process at each bath is Poisson with rate $\varrho$, the distributions for the positions of injections are uniform on $\gamma_L$ an $\gamma_R$, $|\gamma_L|=|\gamma_R|=|\gamma|$, and upon injection, a particle is assigned a random velocity $v$ sampled from the distribution
$$ce^{-m\beta|v|^2}|v| cos(\varphi)dv, \; \; \; c=\frac{2(m\beta)^{3/2}}{\sqrt{\pi}},$$
where $\beta=1/T$, $m$ is a particle's mass, and $\varphi \in (-\frac{\pi}{2},\frac{\pi}{2})$.

\begin{theorem} \label{thm:equilibrium}
The invariant probability measure $\mu$ characterized by the properties below is ergodic.
\begin{itemize}
\item[(i)] The number of particles in the cell is a Poisson random variable with mean
    $$\lambda=2 \sqrt{\pi} \frac{area(\Gamma)}{|\gamma|} \cdot \frac{\varrho \sqrt{m}}{\sqrt{T}},$$
    i.e. $\mu(\Omega_k)=\frac{\lambda^k e^{-\lambda}}{n!}$.
\item[(ii)] $\mu$ has conditional densities $c_k\sigma_k dm_k$ on $\Omega_k$ where $c_k$ is a normalizing constant and for $X \in \Omega_k$,
$$\sigma_k(X)=e^{-\beta(\Theta \sum_{j=1}^{N}\frac{\omega_j^2}{2} + m \sum_{i=1}^{k}\frac{|v_i|^2}{2})}.$$
\end{itemize}
\end{theorem}

\begin{proof}
For $\eta=\frac{\Theta}{m R^2}=1$, the invariance of $\mu$ is shown in \cite{Nonequilibrium}. The argument can be generalized for any $0<\eta<\infty$. Since $\mu \ll m$ and $\mu(S_T)=0$, ergodicity of $\mu$ follows from Theorem \ref{thm:main}. \qed
\end{proof}

The remaining part of the paper is devoted to the proof of Theorem \ref{thm:main}. In order to simplify the exposition, we chose to present the proofs of all the technical lemmas for the situation $\eta=1$ only; all proofs could be generalized for any $0 < \eta < \infty$. Aside from technical lemmas, the argument deals with any $\eta$, $0 < \eta < \infty$.

\section{Outline of Proof of Theorem \ref{thm:main}} \label{sect:outline of proof}

In this section we state three propositions and show that they imply Theorem \ref{thm:main}. Our proof uses some ideas from \cite{Balint}.

Idea of proof: Proposition \ref{prop: abs cont part} states that absolutely continuous measures stay absolutely continuous under $\Phi_\tau$. Propositions \ref{prop: empty playground acquiring density} and \ref{prop:flush particles out} imply that any singular measure eventually acquires an absolutely continuous component when evolved under $\Phi_\tau$. If follows that the singular component of any invariant probability measure with $\mu(S_T)=0$ must be zero. Propositions \ref{prop: empty playground acquiring density} and \ref{prop:flush particles out} also imply that $m$-almost all initial states must belong to the same ergodic component; ergodicity follows by the absolute continuity of $\mu$.

Denote by
\begin{itemize}
\item $S_{SC}$ the set of all states in $\Omega$ such that a simultaneous collision with same disk occurs under the evolution of the system with no particle injections, i.e. for some $t>0$, $q_i^t \in \partial D_k$ and $q_j^t \in \partial D_k$ for $i \ne j$. Note that the evolution of the system is not defined after such a collision.
\item $S_{TC}$ the set of all states in $\Omega$ such that a particle stops under the evolution of the system with no particle injections, i.e. $\exists t$ such that $\forall \tau>t$, $v_i^t=0$. When $\eta=1$, this situation occurs when a particle hits a stopped disk tangentially; when $\eta \ne 1$, it occurs when a particle hits a disk with angular velocity $\omega=\frac{(\eta-1)(v_i)_t}{(\eta+1)R}$ tangentially, where $(v_i)_t$ is a tangential component of the particle's velocity upon collision.
\end{itemize}

Let $S=S_T \cup S_{SC} \cup S_{TC}$.

\begin{definition}
A state $X \in \Omega$ is called admissible if $X \not \in S$.
\end{definition}

Given $t>0$, the probability that no particles are injected on time interval $[0,t]$ is positive. If $\mu$ is an invariant probability measure with $\mu(S_T)=0$, then it cannot give positive measure to either $S_{SC}$ or $S_{TC}$, i.e. $\mu(S)=\mu(S_T \cup S_{SC} \cup S_{TC})=0$.

It follows that if we start with any measure $\nu \ll \mu$, the push forward of $\nu$ under the Markov process $\Phi_\tau$, $(\Phi_\tau)_* \nu$, is well defined for all $\tau>0$.

For any measure $\nu$, denote by $\nu_{\ll}$ and $\nu_\perp$ the absolutely continuous and singular components of $\nu$ with respect to the Lebesgue measure $m$, i.e. $\nu = \nu_{\ll}+\nu_{\perp}$, $\nu_{\ll} \ll m$ and $\nu_\perp \perp m$.

\begin{proposition} \label{prop: abs cont part}
If $\nu \ll m$, then $(\Phi_t)_*\nu$ is well defined for any $t>0$ and $(\Phi_t)_*\nu \ll m$. In particular, $(\Phi_t)_*(\mu_{\ll}) \ll m$ for any $t>0$.
\end{proposition}

Consider a sequence of particle injections $c=(c_1, \cdots, c_n)$ at times $0 < t_1 < t_2 < \cdots < t_n < T$ such that at time $t_i$, a particle enters the system at location $\xi_i \in \gamma_j$, at angle $\delta_i \in (-\frac{\pi}{2},\frac{\pi}{2})$, and with speed $s_i \in (0,\infty)$. Then, assuming no simultaneous collisions with the same disks occur, one can generate a \emph{sample path} $\sigma$ defined on $[0,T]$ in which one starts from state $X \in \Omega$ and injects particles into the system according to $c$.

We call $C$ a \emph{canonical neighborhood} of $c$ if there are disjoint open neighborhoods $T_i$ of $t_i$ contained in $[0,T]$, $\Xi_i$ of $\xi_i$, $\Delta_i$ of $\delta_i$, and $S_i$ of $s_i$ such that for each sequence of injections $c' \in C$, exactly one particle is injected in each $T_i$, with position in $Q_i$, angle in $\Delta_i$, and speed in $S_i$. No other injections occur in the time interval $[0,T]$.

We call $\Sigma$ a \emph{canonical neighborhood} of $\sigma$ if there exist an open neighborhood $U$ of $X$ and a canonical neighborhood $C$ of $c$  such that each sample path in $\Sigma$ starts with an initial condition in $U$, is generated by a sequence of injections from $C$, and no simultaneous collisions with same disks occur. Note that for any $Y \in U$, if $\cal{Y}$ is a set of all sample paths on $[0,T]$ starting at $Y$, then the probability that a sample path from $\cal{Y}$  belongs to $\Sigma$ is positive.

Propositions \ref{prop: empty playground acquiring density} and \ref{prop:flush particles out} split the problem of acquiring density for singular measures in the following way: Proposition \ref{prop: empty playground acquiring density} deals with the simplest situation of acquiring density for a point measure supported on an particle-less initial state $Y_0 \in \Omega_0$; Proposition \ref{prop:flush particles out} provides a sample path from any admissible initial state to a particle-less state, from which a density can be acquired using Proposition \ref{prop: empty playground acquiring density}.

\begin{proposition} \label{prop: empty playground acquiring density}
Given a state $Y_0 \in \Omega_0$ with a condition that if $\eta=1$, all disks have nonzero angular velocities, there exist an open neighborhood $U_0$ of $Y_0$, time $T_0$, and a set $A_0 \subset \Omega_0$ with $m_0(A_0) > 0$, such that for any $Y \in U_0$, $[(\Phi_{T_0})_* \delta_Y]_\ll$ has strictly positive density on $A_0$. In particular, $[(\Phi_{T_0})_* \delta_Y]_\ll (\Omega) \ne 0$.
\end{proposition}

\begin{proposition} \label{prop:flush particles out}
Given an admissible state $X \in \Omega$, a state $Y_0 \in \Omega_0$, and a neighborhood $U_0 \subset \Omega_0$ of $Y_0$, there exist time $T$, a sample path $\sigma$ on $[0,T]$ that starts at $X$ and ends at $Y_0$, and a canonical neighborhood $\Sigma$ of $\sigma$, such that each sample path in $\Sigma$ ends in $U_0$.
\end{proposition}

\begin{theopargself}
\begin{proof}[of Theorem \ref{thm:main} assuming Propositions \ref{prop: abs cont part}, \ref{prop:flush particles out}, and \ref{prop: empty playground acquiring density}]

Propositions \ref{prop: empty playground acquiring density} and \ref{prop:flush particles out} imply that for any admissible state $X \in \Omega$, $\exists$ a neighborhood $U$ of $X$, such that $\forall$ $Y \in U$,
$[(\Phi_{T+T_0})_* \delta_Y]_\ll$ has strictly positive density on $A_0$ and, in particular,
$[(\Phi_{T+T_0})_* \delta_Y]_\ll (\Omega) \ne 0$.

Assume $\mu_{\perp}(\Omega) \not = 0$. Since $\mu$ is invariant with $\mu(S_T)=0$, $\mu_\perp(S)=\mu(S)=0$. Therefore $[(\Phi_{T+T_0}) _* \mu_\perp]_\ll (\Omega) \ne 0$. Applying Proposition \ref{prop: abs cont part} we conclude that $\forall t>T+T_0$, $[(\Phi_t) _* \mu_\perp]_\ll (\Omega) \ne 0$.

Clearly $(\Phi_t)_*\mu=[(\Phi_t)_*(\mu_{\ll})]_{\ll}+[(\Phi_t)_*(\mu_{\ll})]_{\perp}+[(\Phi_t)_*(\mu_\perp)]_{\ll}+[(\Phi_t)_*(\mu_\perp)]_{\perp}$. By Proposition \ref{prop: abs cont part}, $\forall t>0$, $[(\Phi_t)_*(\mu_{\ll})]_{\perp}(\Omega)=0$.  Therefore, for $t>T+T_0$, $$[(\Phi_t)_*\mu]_{\ll}(\Omega)>\mu_{\ll}(\Omega),$$ which contradicts the invariance of $\mu$. This proves the absolute continuity of $\mu$ with respect to the Lebesgue measure $m$.

Assume $\mu_1$ and $\mu_2$ are ergodic measures with $\mu_1(S_T)=\mu_2(S_T)=0$. Then $\mu_i \ll m$, $i=1,2$. Suppose there exists a Borel function $\varphi$ such that $c_1 = \int_\Omega \varphi d\mu_1 \ne \int_\Omega \varphi d\mu_2 = c_2$. Then by the Random Ergodic Theorem, there exist $A_i \in \Omega$, $i=1,2$, with $\mu_i(A_i)=1$ such that for every $X \in A_i$ the ergodic averages are equal to $c_i$ for a.e. sample path starting from $X$; and $m(A_i)>0$, since $\mu_i \ll m$. Since for any admissible state $X \in \Omega$, $[(\Phi_{T+T_0})_* \delta_X]_\ll$ has strictly positive density on $A_0$, the random ergodic averages for all the admissible states are equal for measure $1$ sets of sample paths. Since $m(S)=0$, either $m(A_1)=0$ or $m(A_2)=0$, a contradiction. Therefore $\mu_1=\mu_2=\mu$. \qed
\end{proof}
\end{theopargself}

We prove Proposition \ref{prop: abs cont part} in section \ref{sect:abs cont part}, Proposition \ref{prop: empty playground acquiring density} in section \ref{sect: empty playground acquiring density}, and Proposition \ref{prop:flush particles out} in sections \ref{sect: flush particles out} and \ref{sect: proof flush particles out}.

\section{Proof of Proposition \ref{prop: abs cont part}} \label{sect:abs cont part}
In order to show that, given $\nu \ll m$ and $t >0$, $(\Phi_t)_*\nu$ is defined and $(\Phi_t)_*\nu \ll m$, we would like to consider a countable number of subcases depending on how many particles entered on time interval $(0,t]$ and show that for each subcase the time-$t$ push forward of $\nu$ is well defined and absolutely continuous.

Suppose $n$ particles enter on time interval $(0,t]$. For $n \geq 1$, let $C_n = [(0,t] \times (\gamma_L \cup \gamma_R) \times (-\frac{\pi}{2},\frac{\pi}{2}) \times (0,\infty)]^n /\sim$ be the set of all possible injection parameters for these $n$ particles modulo permutations and let $\rho_n$ be the natural probability measure on $C_n$, i.e. the product of appropriate injection distributions modulo permutations (see subsection \ref{subsect: coupling}). Given $t_1, \cdots, t_{n-1}$ with $0=t_0<t_1< \cdots < t_n=t$, denote by $C_{t_1, \cdots, t_{n-1}}$ the subset of $C_n$ such that $j^{th}$ particle is injected on $(t_{j-1},t_j]$.

For any $\tau > 0$, denote by $(\Phi^n_\tau)_* \nu$ the time-$\tau$ push forward of $\nu$ under the random dynamics assuming that $n$ particles entered on time interval $(0,\tau]$, if defined. Given a subset $C \subset C_n$ of injection parameters, denote by $(\Phi^{n,C}_\tau)_* \nu$ the time-$\tau$ push forward of $\nu$ under the random dynamics assuming that $n$ particles entered on time interval $(0,\tau]$ with injection parameters from $C$, if defined.

\begin{lemma} \label{lemma: prop 1 decomposition}
If $\nu \ll m$, $(\Phi^0_t)_*\nu$ is defined and absolutely continuous with respect to $m$. Moreover, for any $n \geq 1$ and a choice of $t_1, \cdots, t_{n-1}$ with $0=t_0<t_1< \cdots < t_n=t$, $(\Phi^{n,C_{t_1, \cdots, t_{n-1}}}_t)_* \nu$ is defined and absolutely continuous with respect to $m$.
\end{lemma}

\begin{theopargself}
\begin{proof}[of Proposition \ref{prop: abs cont part} assuming Lemma \ref{lemma: prop 1 decomposition}]

Let $B=\{X \in \Omega: P^t_X$ is not defined $\}$. Then for each $X \in B$, there exists a positive measure set $C_X$ of injections such that, if we start at $X$ and follow any sequence of injections $c \in C_X$, the dynamics is not defined up to time $t$, i.e. a simultaneous collision with same disk occurs before time $t$ for each such injection.

Assume $m(B)>0$. Define $A=\{(X,c):X \in B; \; c \in C_X \}$. Then by Fubini's theorem $(m \times \rho_n) (A)>0$. Let $(X,c)$ be a Lebesgue density point of $A$ with $X$ admissible and $c$ having no simultaneous injections ($(X,c)$ exists since $m(S)=0$ and probability of simultaneous injections is zero). If $t_{c_1}, \cdots, t_{c_n}$ are the injection times for $c$, choose any $t_1, \cdots, t_{n-1}$ such that $0=t_0<t_{c_1}<t_1< \cdots <t_{n-1}<t_{c_n}<t_n=t$ ($c$ cannot be an empty sequence of injections since $(\Phi^0_t)_*\nu$ is defined by Lemma \ref{lemma: prop 1 decomposition}). Let $U$ be any neighborhood of $X$ such that any $Y \in U$ is admissible (the set of admissible states is open and dense in $\Omega$). Then $(m \times \rho_n) [A \cap (U \times C_{t_1, \cdots, t_{n-1}})]>0$, while by Lemma \ref{lemma: prop 1 decomposition} it must be zero. A contradiction. Therefore $m(B)=0$ and $(\Phi^n_t)_* \nu$ is defined.

Suppose $(\Phi^n_t)_* \nu$ has a nonzero singular component supported on a set $F \subset \Omega$. Then $m(F)=0$ and $m\{X \in \Omega: P^t_X(F)>0\}>0$. By an argument similar to the above we get a contradiction to Lemma \ref{lemma: prop 1 decomposition}.
Therefore $(\Phi^n_t)_* \nu$ is defined and absolutely continuous with respect to $m$, which completes the proof of Proposition \ref{prop: abs cont part}. \qed
\end{proof}
\end{theopargself}

\begin{theopargself}
\begin{proof}[of Lemma \ref{lemma: prop 1 decomposition}]

Since $C_{t_1, \cdots, t_{n-1}} \subset C_n$ is such that one particle is injected on each $(t_{j-1},t_j]$ and the probability of two particles entering at the same time is zero, to prove Lemma \ref{lemma: prop 1 decomposition} it is enough to treat two situations: $0$ particles are injected on $(0,t]$ and 1 particle is injected on $(0,t]$. \medbreak

\begin{case}[(0 particles injected)]
Since $m(S)=0$, where $S$ is a set of non-admissible states, and $\nu \ll m$, the time-$\tau$ push forward of $\nu$ is well defined in this situation.

The system behaves as a closed one until some particles exit. We would like to split the possible situations according to the number of particles that exit:

\begin{description}
\item [0:] \emph{no particles exit.} \\ Let $A_0 \subset \Omega$ be the set of states such that no particles exit on time interval $(0,t]$. Decompose $A_0=\cup A_0^k$ such that $A_0^k=A_0 \cap \Omega_k$. Then $\nu|_{A_0^k} \ll m_k$ and must stay so at all times $(0,t]$, i.e. $(\Phi_\tau^0)_*(\nu|_{A_0^k}) \ll m_k$ $\forall \tau \in (0,t]$. Indeed, if for some $\tau \leq t$ and a Borel set $B$, $[(\Phi_\tau^0)_*(\nu|_{A_0^k})]_\perp(B) \not = 0$ but $m_k(B)=0$, then $(\nu|_{A_0^k})(\Phi_{-\tau}B) \geq [(\Phi_\tau^0)_*(\nu|_{A_0^k})]_\perp(B)  \not = 0$, while $m_k(\Phi_{-\tau}B)=m_k(B)=0$, a contradiction to the absolute continuity of $\nu|_{A_0^k}$. Thus $(\Phi_t^0)_*(\nu|_{A_0}) \ll m$.

\item [1:] \emph{1 particle exits or several particles exit at the same time.} \\ Let $A^\tau_1$ be the set of states such that $1$ particle exits on the time interval $(\tau,t]$ and does not collide with $\partial \Gamma \setminus (\gamma_L \cup \gamma_R)$ on that time interval. If several particles exit at the same time, let $A^\tau_1$ be the set of states such that neither of these particles collides with $\partial \Gamma \setminus (\gamma_L \cup \gamma_R)$ on the time interval $(\tau,t]$.
    Then until time $\tau$, the dynamics is equivalent to the dynamics of the closed system and $(\Phi_\tau^0)_*(\nu|_{A_1^\tau}) \ll m$. Since after the time $\tau$ the particle(s) do not collide with the disk, the particle(s) coordinates are independent from the rest of the system and  $(\Phi_t^0)_* (\nu|_{A_1^\tau})=(\Phi_{t-\tau}^0)_*(P[(\Phi_\tau^0)_* (\nu|_{A_1^\tau})]) \ll m$, where $P$ denotes the projection of the measure on the remaining coordinates in the system. Since the statement is true for any $\tau \in [0,t)$, $(\Phi_t^0)_*(\nu|_{A_1}) \ll m$, where $A_1=\cup_{\tau \in [0,t)} A_1^\tau$.

\item [n:] \emph{particles exit at $n$ different times with possibly several exiting at each time.} \\  Let $A_n$ be the set of states such that $n$ particle exits occur on the time interval $(0,t]$, with possibly several particles exiting at the same time. Then $$A_n=\cup_{0=t_0< t_1 < \cdots < t_n=t} A_{t_1, \cdots, t_{n-1}}, $$ where $A_{t_1, \cdots, t_{n-1}}$ denotes the set of states such that $j^{th}$ particle exists on the time interval $(t_{j-1},t_j]$, $1 \leq j \leq n$. Each $A_{t_1, \cdots, t_{n-1}}$ can be treated as in the previous situation by applying the argument consequently to each time interval $(t_{j-1},t_j]$, $1 \leq j \leq n$. Thus $(\Phi_t^0)_*(\nu|_{A_n}) \ll m$.
\end{description}

We conclude that $(\Phi_t^0)_* \nu \ll m$.
\end{case}

\begin{case}[(1 particle injected)]
In order treat this case, we need to study the injection process in detail. We would like to start with a simplified situation when the injected particle does not collide with disks or walls on time interval $(0,t]$. We start by studying the injection process for the cell $\tilde{\Gamma}$ occupying the $x>0$ half plane of $xy$-plane bounded by the vertical wall along the $y$-axis with an opening $(-\gamma/2,\gamma/2)$ along the $y$-axis.

\begin{lemma}[(Particles are injected with 4-dimensional uncertainty)] \label{lemma: particle enters 4D}
Let $\tilde{\Gamma}$ be as above. Assume a particle is injected through the opening at a random entrance time $\tau \in (0,\infty)$, with a random position $\xi \in (\-\gamma/2,\gamma/2)$, at a random angle $\delta \in (-\pi/2,\pi/2)$ and with a random speed $s \in (0,\infty)$. Assume also that the distributions for entrance time, position, angle and speed are finite and positive on all the given intervals. Then the measure that describes the probability of finding the particle at certain position $(x^t,y^t)$ with certain velocity $(v_x^t,v_y^t)$ at any time $t>0$ is absolutely continuous with respect to the Lebesgue measure on $\tilde{\Gamma} \times \mathbb{R}^2$ and has positive density everywhere.
\end{lemma}

\begin{proof}[of Lemma \ref{lemma: particle enters 4D}]
The lemma holds if the mapping $f^t:(0,\infty) \times (\-\gamma/2,\gamma/2)\times (-\pi/2,\pi/2) \times (0,\infty) \to (0,\infty) \times (-\infty,\infty) \times (0,\infty) \times (-\infty, \infty)$ such that $(\tau,\xi,\delta,s) \to (x^t,y^t,v_x^t,v_y^t)$ is a diffeomorphism. This boils down to verifying that the Jacobian determinant
is nonzero everywhere and that the map is surjective.
Simple computations yield that:
$x^t=(t-\tau)s \cos(\delta)$, $y^t=-(t-\tau)s \sin(\delta)+\xi$, $v_x^t=s \cos(\delta)$, and $v_y^t=-s \sin(\delta)$ and thus the Jacobian matrix is:
$$\left(
    \begin{array}{cccc}
      -s\cos(\delta) & 0 & -(t-\tau)s\sin(\delta) & (t-\tau)\cos(\delta) \\
      s\sin(\delta) & 1 & -(t-\tau)s\cos(\delta) & -(t-\tau)\sin(\delta) \\
      0 & 0 & -s\sin(\delta) & \cos(\delta) \\
      0 & 0 & -s\cos(\delta) & -\sin(\delta) \\
    \end{array}
  \right)
$$
and its determinant is $-s^2\cos(\delta)$, which is nonzero on the domain of definition of $f^t$.

Given $(x^t,y^t,v_x^t,v_y^t)$, let
$$s=\sqrt{(v_x^t)^2+(v_y^t)^2}, \; \; \; \delta=\arctan(-v_y^t/v_x^t),$$ $$\xi=\frac{x^t\sin(\delta)+y^t\cos(\delta)}{\cos(\delta)}, \; \; \; \tau=\frac{y^t-\xi+ts\sin(\delta)}{s\sin(\delta)}.$$ Such $(\tau,\xi,\delta,s)$ clearly maps to $(x^t,y^t,v_x^t,v_y^t)$. \qed
\end{proof}

In Lemma \ref{lemma: particle enters 4D} the domain $\tilde{\Gamma}$ is such that an injected particle cannot collide with the boundary of $\tilde{\Gamma}$ at any time. In the real situation, when a particle is injected into the playground $\Gamma$, it might collide with $\partial \Gamma$ in an arbitrarily short time, depending on the injection parameters.

Let $C_1^\kappa$ be the set of all particle injection parameters $(\tau,\xi,\delta,s)$ such that one particle enters on time interval $(0,\kappa]$, $0 < \kappa \leq t$, and does not collide with $\partial \Gamma$ on time interval $(0, \kappa]$; no other particles are injected on time interval $(0,t]$. Since we assumed that all injection distributions have positive density everywhere and $C_1^\kappa$ clearly has a nonempty interior, by Lemma \ref{lemma: particle enters 4D} the measure that describes the particle location and velocity at time $\kappa$ is absolutely continuous with respect to the Lebesgue measure on $\Gamma \times \mathbb{R}^2$; denote this measure by $\pi_{C_1^\kappa}$.

Suppose we start with a measure $\nu \ll m$ and assume that one particle is injected on $(0,t]$ with injection parameters drawn from $C^1_\kappa$. Then until time $\kappa$ the coordinates associated with the injected particle are uncoupled from the coordinates of the rest of the system. Thus $\nu$ evolves as in Case 1, $(\Phi_\kappa^0)_* \nu$ is defined, and $(\Phi_\kappa^0)_* \nu \ll m$.

Therefore
$$(\Phi_\kappa^{1, C_1^\kappa})_*\nu = [\pi_{C_1^\kappa}] \times [(\Phi_\kappa^0)_* \nu] \ll m.$$
The injected particle have entered by time $\kappa$ and no other particles enter on time interval $(0,t]$. Thus on time interval $(\kappa,t]$ the system behaves as in Case $1$, implying that $(\Phi_t^{1,C_1^\kappa})_*\nu \ll m$. Since the above result is true for any $\kappa \in (0,t]$, $(\Phi_t^1)_* \nu \ll m$. \qed
\end{case}
\end{proof}
\end{theopargself}

\section{Proof of Proposition \ref{prop: empty playground acquiring density}} \label{sect: empty playground acquiring density}

We are going to \textquotedblleft acquire density" starting with an initial state $Y_0 \in \Omega_0$ by hitting each disk once with a particle possessing a \textquotedblleft 4-dimensional uncertainty". Upon each such collision, each disk would acquire a \textquotedblleft 2-dimensional uncertainty." That is, we will show that there exists an open set of injections $C$ of $N$ particles, such that the $j^{th}$ particle hits disk $D_j$ once and exits the system with no additional collisions by some uniformly selected time $T_0$. Moreover, $(\Phi^{N,C}_{T_0})_* \delta_{Y_0} \ll m$. It will turn out that same statement is true for nearby states $Y \in \Omega_0$ with $(\Phi^{N,C}_{T_0})_* \delta_Y$ varying continuously with $Y$.

In order to prove Proposition \ref{prop: empty playground acquiring density} we need the following lemmas:

\begin{lemma} [Continuity] \label{lemma:cts dependence of particle and disk's position and velocity}
Let $B \subset \Gamma \times \mathbb{R}^2$ be an open ball of positions and velocities such that for any $(x,y,v_x,v_y) \in B$ a particle with position $(x,y)$ and velocity $(v_x,v_y)$ at time $t=0$ it is going to hit disk $D_j$ within time $\tau_0$ while not meeting $\partial \Gamma$ again on the time interval $[0,\tau_1]$, $\tau_1>\tau_0$. Note that the collision with $D_j$ must be non-tangential since $B$ is open. Let $F_j$ be an open ball of angular positions and velocities of disk $D_j$ at time $0$. For any $t \in (\tau_0,\tau_1)$, define $f_t:B \times F_j \to \Gamma \times \mathbb{R}^2 \times \partial D_j \times \mathbb{R}$ such that $f_t(x,y,v_x,v_y,\varphi,\omega)= (x^t,y^t,v_x^t,v_y^t, \varphi^t,\omega^t)$ gives the positions and the velocities of the particle and disk $D_j$ at time $t$. Then $f_t$ is continuous for any $t \in (\tau_0,\tau_1)$.
\end{lemma}

\begin{lemma} [Acquiring density for a disk] \label{lemma: 4D particle hits disk}
Let $B \subset \Gamma \times \mathbb{R}^2$ be as in Lemma \ref{lemma:cts dependence of particle and disk's position and velocity}. Assume at time $0$ disk $D_j$ has position $\varphi$ and angular velocity $\omega$. Define $\Psi_t^{\varphi,\omega}:B \to S^1 \times \mathbb{R}$, $(x,y,v_x,v_y) \to (\varphi^t,\omega^t)$ to be the mapping of the particle position and velocity at time zero to the disk position and velocity at time $t \in (\tau_0,\tau_1)$. Let $\nu$ be a measure on $B$ equivalent to the Lebesgue. Then the push forward measure $(\Psi_t^{\varphi,\omega})_*\nu$ is absolutely continuous with respect to the Lebesgue measure on $S^1 \times \mathbb{R}$ and has positive density on some open set.
\end{lemma}

\begin{theopargself}
\begin{proof}[of Proposition \ref{prop: empty playground acquiring density} assuming Lemmas \ref{lemma:cts dependence of particle and disk's position and velocity} and \ref{lemma: 4D particle hits disk}]

Denote the angular positions and angular velocities of the disks in state $Y_0$ by $(\varphi_1,\omega_1)$, $\cdots$, $(\varphi_N,\omega_N)$. By assumption, if $\eta=1$, $\omega_j \ne 0$, $1 \leq j \leq N$.

Suppose we inject a particle at time $\tau>0$ with some initial position $\xi \in \gamma_L$, initial angle $\varphi \in (-\frac{\pi}{2},\frac{\pi}{2})$, and speed $s \in (0,\infty)$ arranged in such a way that it first collides non-tangentially with $D_j$ at the top, $(2j-1,R)$, and exits the system with no additional collisions; this is possible because $\omega_j \ne 0$, $1 \leq j \leq N$, if $\eta=1$. Since the collision with disk in non-tangential, by Lemmas \ref{lemma: particle enters 4D} and \ref{lemma:cts dependence of particle and disk's position and velocity} there exist open neighborhoods $I_j$ of $\omega_j$ and $C_j$ of $(\tau,\xi,\varphi,s)$ such that for each $\omega' \in I_j$ and $(\tau,\xi,\varphi,s)' \in C_j$, the injected particle follows a nearby path in $\Gamma$, hits $D_j$ once, and exits the system with no additional collisions within some time $T_j$.

Suppose we subsequently hit each disk with a particle as described above. Define a neighborhood $U$ of $Y_0$ by $U = \partial D_1 \times \cdots \times \partial D_N \times I_1 \times \cdots \times I_N$.

Let $C=\Pi_{j=1}^N C_j$ and $T=\sum_{j=1}^{N} T_j$. Then by Lemmas \ref{lemma: particle enters 4D} and \ref{lemma: 4D particle hits disk}, for any $Y \in U$, $(\Phi^{N,C}_T)_*\delta_{Y} \ll m_0$. Here $(\Phi^{N,C}_T)_*\delta_{Y}$ denotes the time-$T$ push forward of $\delta_{Y}$, provided that exactly $N$ particles enter on time interval $(0,T]$ allowing only injections with parameters in $C$. Let $A_{Y}$ be the set of states on which $(\Phi^{N,C}_T)_*\delta_{Y}$ has strictly positive density. By Lemma \ref{lemma: 4D particle hits disk}, each $A_Y$ contains an open set and by Lemma \ref{lemma:cts dependence of particle and disk's position and velocity}, $A_Y$ vary continuously with $Y$. Therefore there exists an open neighborhood $U_0$ of $Y_0$, $U_0 \subset U$, such that $A_0=\cap_{Y \in U_0} A_Y$ contains an open set; clearly $m_0(A_0)>0$. This completes the proof of Proposition \ref{prop: empty playground acquiring density}. \qed
\end{proof}

\begin{proof}[of Lemmas \ref{lemma:cts dependence of particle and disk's position and velocity} and \ref{lemma: 4D particle hits disk} ($\eta=1$)]

Assume the disk is of radius $R$ and is centered at $(0,0)$ coordinate of the $xy$-plane and the particle and the disk have initial coordinates $(x,y,v_x,v_y)$ and $(\varphi,\omega)$ respectively.

Denote by $\tau$ the collision time with the disk, by $\theta$ the angular position of the collision point on the disk measured counterclockwise from the positive direction of the $x$-axis, and by $v_t$ and $v_\perp$ the tangential and the normal velocities of the particle upon collision, with $v_\perp$ representing the velocity after the collision, i.e. pointing outwards. Then,
$$v_t=-v_y\cos(\theta)+v_x\sin(\theta)$$
$$v_\perp=-v_y\sin(\theta)-v_x\cos(\theta)$$
$$\tau=\frac{R\cos(\theta)-x}{v_x}=\frac{R\sin(\theta)-y}{v_y}$$
and
$$v_t^2+v_\perp^2=v_x^2+v_y^2$$
One can rewrite
$$Rv_t=v_xy-v_yx$$
$$Rv_\perp=-\tau(v_x^2+v_y^2)-v_yy-v_xx=\sqrt{v_x^2+v_y^2-v_t^2}$$

Following the collision, at time $t \in (\tau_0,\tau_1)$, the coordinates for position and velocity of the particle are:
$$v_x^t=v_\perp(\tau v_x + x)/R + \omega (\tau v_y + y)$$
$$v_y^t=v_\perp(\tau v_y + y)/R - \omega (\tau v_x + x)$$
$$x^t=\tau v_x + x + v_x'(t-\tau)$$
$$y^t=\tau v_y + y + v_y'(t-\tau)$$
The angular velocity of the disk at time $t$ is $w^t=v_t/R$ and the angular position is $\varphi^t=[\varphi+\omega\tau+v_t(t-\tau)/R] \mod 2\pi R$. Clearly $f_t$ is continuous, which completes the proof of Lemma \ref{lemma:cts dependence of particle and disk's position and velocity}.

In order to prove Lemma \ref{lemma: 4D particle hits disk}, we would like to study the matrix of the derivatives of $\Psi^{\varphi,\omega}_t:B \to S^1 \times \mathbb{R}$, $(x,y,v_x,v_y) \to (\varphi^t,\omega^t)$, $t \in (\tau_0,\tau_1)$. It is:
$$\left(
    \begin{array}{cc}
      \frac{\partial \varphi^t}{\partial x} & \frac{\partial \omega^t}{\partial x} \\
      \frac{\partial \varphi^t}{\partial y} & \frac{\partial \omega^t}{\partial y} \\
      \frac{\partial \varphi^t}{\partial v_x} & \frac{\partial \omega^t}{\partial v_x} \\
      \frac{\partial \varphi^t}{\partial v_y} & \frac{\partial \omega^t}{\partial v_y}\\
    \end{array}
 \right)^T
=\left(
     \begin{array}{cc}
       \frac{(-\frac{v_t v_y}{v_\perp}-v_x)( \omega-\frac{v_t}{R})}{v_x^2+v_y^2}-\frac{v_y}{R^2}(t-\tau) & -\frac{v_y}{R^2} \\
       \frac{(\frac{v_t v_x}{v_\perp}-v_y)( \omega-\frac{v_t}{R})}{v_x^2+v_y^2}+\frac{v_x}{R^2}(t-\tau) & \frac{v_x}{R^2} \\
       \frac{(\frac{v_t y}{v_\perp}-2v_x\tau-x)( \omega-\frac{v_t}{R})}{v_x^2+v_y^2}+\frac{y}{R^2}(t-\tau) & \frac{y}{R^2} \\
       \frac{(-\frac{v_t x}{v_\perp}-2v_y \tau-y)( \omega-\frac{v_t}{R})}{v_x^2+v_y^2}-\frac{x}{R^2}(t-\tau) & -\frac{x}{R^2} \\
     \end{array}
   \right)^T,
$$
where $T$ denotes the transpose.

Suppose the two rows of the derivative matrix (i.e. columns of the non-transposed matrix) are linearly dependent, then there exists a constant $c$ such that:
$$c=
-\frac{R^2}{v_y}\frac{(-\frac{v_t v_y}{v_\perp}-v_x)( \omega-\frac{v_t}{R})}{v_x^2+v_y^2}+(t-\tau)=
\frac{R^2}{v_x}\frac{(\frac{v_t v_x}{v_\perp}-v_y)( \omega-\frac{v_t}{R})}{v_x^2+v_y^2}+(t-\tau)=$$
$$=\frac{R^2}{y}\frac{(\frac{v_t y}{v_\perp}-2v_x\tau-x)( \omega-\frac{v_t}{R})}{v_x^2+v_y^2}+(t-\tau)= -\frac{R^2}{x}\frac{(-\frac{v_t x}{v_\perp}-2v_y \tau-y)(\omega-\frac{v_t}{R})}{v_x^2+v_y^2}+(t-\tau)$$
Thus either $(R \omega - v_t)=0$ or
$$\frac{v_t}{v_\perp}+\frac{v_x}{v_y}=\frac{v_t}{v_\perp}-\frac{v_y}{v_x}=
\frac{v_t}{v_\perp}-\frac{2v_x\tau}{y}-\frac{x}{y}=\frac{v_t}{v_\perp}+\frac{2v_y\tau}{x}+\frac{y}{x}$$
$$\Leftrightarrow$$
$$x=y=v_x=v_y=0,$$
From our assumptions it follows that $x,y,R,v_x,v_y \not = 0$; so the derivative matrix has rank $2$ unless $R \omega = v_t$.

Let $A_\omega = \{(x,y,v_x,v_y): R \omega=v_t=\frac{1}{R}(v_xy-v_yx)\}$. Clearly $\nu(A_\omega)=0$. For each point $(x,y,v_x,v_y) \in B \setminus A_\omega$, the derivative matrix has rank 2. Therefore the push forward of $\nu$ under $\Psi^{\varphi, \omega}_t$ must be absolutely continuous with respect to the Lebesgue measure on $S^1 \times \mathbb{R}$ and has positive density on an open set $\Psi^{\varphi,\omega}_t(B \setminus A_\omega)$. This completes the proof of Lemma \ref{lemma: 4D particle hits disk}. \qed
\end{proof}
\end{theopargself}

\section{Flushing Particles Out} \label{sect: flush particles out}

This section is the $1^{st}$ step in the proof of Proposition \ref{prop:flush particles out}. Here we will show the following:

\begin{proposition} \label{prop:sample path flush particles out}
For any admissible state $X$, there exists a sample path $\sigma_X$ that starts at $X$ and ends at some $X_0 \in \Omega_0$.
\end{proposition}

In order to drive the system from state $X$ with possibly many particles to some particle-less state $X_0$, we have to ensure that each particle in $X$ traces a path in $\Gamma$ from its initial position to one of the exits. Following the ideas in \cite{Eckmann}, we will describe a class of projected particle paths traced in $\Gamma$ and show that each can be followed provided that disks have appropriate angular velocities upon collisions. Then we will establish that by injecting particles with appropriate initial conditions we can change the angular velocity of any disk to any given value in an arbitrarily short time. That will enable us to force a particle along a projected particle path  by setting the angular velocities of the disks to appropriate values before collisions.

\begin{definition}
A \emph{proper projected particle path} is a continuous curve $\gamma:[0,1] \to \Gamma$, $s \mapsto \gamma(s)$, such that
\begin{enumerate}
\item $\gamma$ consists of a finite sequence of straight segments meeting at $\partial \Gamma$.
\item The incoming and outgoing angles of two consecutive segments of $\gamma$ meeting $\partial \Gamma_0$ are equal.
\item Only $\gamma(0)$ and $\gamma(1)$ can be in the openings $\gamma_L$ and $\gamma_R$.
\item $\gamma$ is nowhere tangent to boundaries of the disks $\cup_{j=1}^{N} \partial D_j$.
\end{enumerate}
\end{definition}

\begin{remark}
Note that a proper projected particle path is allowed to have \emph{any} non-tangential \textquoteleft reflections' off the boundaries of the disks $\cup_{j=1}^{N} \partial D_j$. An example of a proper projected particle path is shown in Fig \ref{fig:rectangle}.
\end{remark}

\begin{lemma} [Existence of a proper projected particle path] \label{lemma:proper projected particle path}
There exits a proper projected particle path from any point $(x,y) \in \partial D_j$ to one of the exits $\gamma_L$ or $\gamma_R$.
\end{lemma}

The statement in Lemma \ref{lemma:proper projected particle path} is rather obvious for the geometry we consider. We include a proof for completeness purposes.

\begin{theopargself}
\begin{proof}[of Lemma \ref{lemma:proper projected particle path}]
Let $\delta$ be such that $0<\delta<\frac{2R\sqrt{1-R}}{2-R}$. If $y$-coordinate of $(x,y)$ is not in $[-\delta,\delta]$, i.e. $|y|>\delta$, then by a simple geometric argument there exists a proper projected particle path from $(x,y)$ to the top  of the disk $D_j$, $(2j-1, R)$ (or the bottom $(2j-1, -R)$), which makes several collisions with the upper wall $[0,N] \times \{1\}$ (or the lower wall $[0,N] \times \{- 1\}$). By appending to the proper projected particle path above a segment in $\Gamma$ that connects $(2j-1, \pm k)$ to either $\gamma_L$ or $\gamma_R$ , we prove Lemma \ref{lemma:proper projected particle path} for the situation when $|y|>\delta$. If $y \in [-\delta, \delta]$, then there exits a segment in $\Gamma$ that connects $(x,y)$ to the point on the appropriate nearby disk or exit with $y$-coordinate between $\delta$ and $\frac{2R\sqrt{1-R}}{2-R}$. Indeed, by the system's geometry, a tangent line from $(2j-1 \pm R, 0)$ to the appropriate nearby disk (if applicable) intersects the disk with $y$-coordinate $\pm\frac{2R\sqrt{1-R}}{2-R}$. \qed
\end{proof}
\end{theopargself}

In order to force a particle to follow a proper projected particle path, we have to ensure that upon each (non-tangential) collision with a disk, the disk has appropriate angular velocity: the total velocity of the particle after collision is the vector sum of $\vec{v}_\perp'=-\vec{v}_\perp$ and $\vec{v}_t'=\vec{v}_t -\frac{2\eta}{1+\eta}(\vec{v}_t - R \vec{\omega})$ and must be parallel to the next segment of the proper projected particle path. We would like to establish that by injecting particles with appropriate initial conditions we can change the angular velocity of any disk to any given value in an arbitrarily short time.

\begin{lemma} [Controlling angular velocities of disks in arbitrarily short times] \label{lemma: controllability}
Suppose disk $D_j$ rotates with angular velocity $\omega$ and none of the particles inside the system will collide with any disk before time $\tau>0$. Given any $\omega'$, there exists a sequence of particle injections on time interval $(0,\tau)$ from the left bath such that:
\begin{itemize}
\item at time $\tau$ the disk $D_j$ has angular velocity $\omega'$,
\item at time $\tau$ all the injected particles have left the system, and
\item on time interval $(0,\tau)$ the injected particles follow admissible paths and only hit disks $D_1, \cdots, D_{j-1}$ with the exception of one collision of one particle with disk $D_j$.
\end{itemize}
The same holds for the right bath with appropriate disk renumbering.
\end{lemma}

We prove the Lemma \ref{lemma: controllability} in subsection \ref{subsect:main lemma proof}.

\subsection{Proof of Proposition \ref{prop:sample path flush particles out}: no tangential collisions} \label{subsect: flush particles out: no tang collisions}

By definition of an admissible state, one of the following holds for each particle in $X$ under the evolution of the system with no particle injections:
\begin{enumerate}
\item there exists a finite time $t>0$ such that the particle exits the system at time $t$ and does not collide with any disks on time interval $[0,t]$, i.e. $q^{t} \in \gamma_L \cup \gamma_R$ and $q^{\tau} \not \in \cup_{j=1}^N \partial D_j$ $ \forall \tau \in [0,t]$;
\item there exists a finite time $t>0$ such that at time $t$ the particle collides with a disk non-tangentially and no other disk collisions occur on time interval $[0,t]$, i.e. $q^{t} \in \cup_{j=1}^N \partial D_j$ with $v^{t}_\perp \not = 0$ and $q^{\tau} \not \in \cup_{j=1}^N \partial D_j$ $\forall \tau \in [0,t]$;
\item there exists a finite time $t>0$ such that at time $t$ the particle collides with a disk tangentially with $v_x \ne 0$ and no other disk collisions occur on time interval $[0,t]$, i.e. $q^{t} \in \cup_{j=1}^N \partial D_j$ with $v^{t}_\perp = 0$ and $v_x \ne 0$ and $q^{\tau} \not \in \cup_{j=1}^N \partial D_j$ $\forall \tau \in [0,t]$.
\end{enumerate}

In this subsection we are going to assume that for each particle in $X$ either 1 or 2 holds, i.e. under the evolution of the system with no particle injections each particle in $X$ either exits the system or collides with a disk non-tangentially. We will treat the situation with tangential collisions in subsection \ref{subsect:flush particles out: tang collisions}.

Suppose $X$ contains $k$ particles and the above assumption is satisfied. Then using Lemma \ref{lemma:proper projected particle path} for the $j^{th}$ particle in $X$, $1 \leq j \leq k$, we can assign a proper projected particle path $\gamma_j$ from the particle's initial position $(q_j^0,v_j^0)$ to one of the exits. If we force the particle to follow this path, the times of all collisions are fixed. Indeed, the unique angular velocity of the disk keeps the particle on the path for each (non-tangential) collision, implying that the speeds with which the particle traces segments of the path and therefore the collision times are uniquely determined from the path. If $j^{th}$ particle follows $\gamma_j$, let $\tau^j_1, \cdots, \tau^j_{n(j)}$ be the times of collisions with disks, $D_{k(j,1)}, \cdots, D_{k(j,n(j))}$ and let $\omega^j_1, \cdots, \omega^j_{n(j)}$ be the required angular velocities.

Assume first that all $\tau^j_i$ (for all particles in $X$ and all collisions) are different. Then direct application of the Lemma \ref{lemma: controllability} between collisions guarantees the existence of a sample path $\sigma_X$ from state $X$ to some state $X_0 \in \Omega_0$.

If some $\tau^j_i$ happen to coincide, simultaneous collisions with same disks might occur, making us unable to construct $\sigma_X$ with our choice of paths. We would like to show that we can always choose a collection of nearby paths from the particles' initial positions to the exits such that no simultaneous collisions with same disks occur. Note that we cannot, in general, avoid all simultaneous disk collisions since we have no control over the times of the first disk collisions of the particles in $X$. First collisions, however, cannot occur simultaneously with same disks since $X$ is admissible.

We treat the possibility of simultaneous collisions with different disks as follows: If several, say $n$, collisions are about to occur in time $\tau$ after the previous collision time, all with different disks, we can set the angular velocities of the disks by the Lemma \ref{lemma: controllability} in the order of decreasing disk index, in a fraction of time $\tau$, $\frac{\tau}{n}$, each: the Lemma \ref{lemma: controllability}  guarantees that the disks with indexes larger than the one we set the angular velocity for are left untouched.

In the remaining part of the proof we show that we can choose a collection of paths for particles in $X$ from their initial positions to the exits such that no simultaneous collisions occur. We start with a description of a set of paths to choose from for each particle in $X$. In the following Lemma, we assume that at time $0$ the $j^{th}$ particle is the only particle in the system in order to ensure that the system is defined at all times.

\begin{lemma} \label{lemma: open set of paths}
Assume that under the evolution of the system with no particle injections the $j^{th}$ particle collides with a disk non-tangentially in some finite time $t>0$. Then there exist open neighborhoods $I^j_1$ of $\omega^j_1$, $\cdots$, $I^j_{n(j)}$ of $\omega^j_{n(j)}$ such that for any choice of angular velocities $(\omega^j_1)' \in I^j_1$, $\cdots$, $(\omega^j_{n(j)})' \in I^j_{n(j)}$, if a particle starts with its initial position $(q^j_0,v^j_0)$, it collides with $D^j_{k(j,1)}, \cdots, D^j_{k(j,n(j))}$ set to $(\omega^j_1)'$, $\cdots$, $(\omega^j_{n(j)})'$, and exits the system.
\end{lemma}

Lemma \ref{lemma: open set of paths} follows from Lemma \ref{lemma:cts dependence of particle and disk's position and velocity} and the continuity of the billiard flow.

Lemma \ref{lemma: open set of paths} guarantees that the $j^{th}$ particle exits the system for any choice of angular velocities from $I^j_1$, $\cdots$, $I^j_{n(j)}$. If a particle starts at $(q^j_0,v^j_0)$, for different choices of angular velocities in $I^j_1$ for disk $D_{k(j,1)}$, the particle collides with disk $D_{k(j,2)}$ at different times. In fact, one gets an open set of possible collision times since the set of possible positions of the particle at any time $\tau$ from the collision with disk $D_{k(j,1)}$ forms a broken line: before any collisions, the particle's positions are $$\tau \vec{v}'=\tau (\vec{v}_\perp'+ \vec{v}_t')=
\tau (\vec{v}_\perp'+[\vec{v}_t-\frac{2\eta}{1+\eta}(\vec{v}_t - R (\vec{\omega}^j_1)')])=$$
$$=\tau (\vec{v}_\perp'+[\vec{v}_t-\frac{2\eta}{1+\eta}(\vec{v}_t - R \vec{\omega}^j_1)])+
\frac{2\eta R \tau }{1+\eta}(\vec{\omega}^j_1-(\vec{\omega}^j_1)'),$$
where $(\omega^j_1)'$ varies through $I^j_1$; and upon reflections from straight walls the straight line of positions becomes a broken line.  If we fix a specific $(\omega^j_1)' \in I^j_1$, similar argument shows that different $(\omega^j_2)' \in I^j_2$ for disk $D_{k(j,2)}$ yield a range of collision times with $D^j_{k(j,3)}$; and so on.

Therefore, can always pick a collection of $(\omega^j_i)' \in I^j_i$ such that no simultaneous collisions occur for all particles in the system. Given such a choice, each particle in $X$ follows a path from its initial position to an exit. Setting the angular velocities of the disks to $(\omega^j_i)'$ at appropriate times by following appropriate sequence of injections provided by the Lemma \ref{lemma: controllability}, we construct a sample path $\sigma_X$ from state $X$ to some state $X_0$ in $\Omega_0$. This completes the proof of Proposition \ref{prop:sample path flush particles out} for the initial states $X$ in which particles either hit a disk non-tangentially or exit the system. \qed

\subsection{Tangential Collisions} \label{subsect:flush particles out: tang collisions}
In this section we will prove Proposition \ref{prop:sample path flush particles out} for the general case, i.e. when for some particles in $X$ the first collision with a disk might be tangential. Since $X$ is admissible, such collisions must occur with $v_x \ne 0$ and without simultaneous collisions with same disks.

Suppose the first collision of the $j^{th}$ particle is tangential with $v_x \ne 0$. By the Lemma \ref{lemma: controllability} we can ensure that, upon collision, $R$ times the angular velocity of the disk is equal to the velocity of the $j^{th}$ particle, $R \omega=v=v_t$. After such a collision the particle continues along the straight line with the same velocity as if no collision has occurred. From that point, the particle will either exit the system, hit a disk non-tangentially, or hit a disk tangentially again. In the third situation we would like to set the angular velocity of the disk to be equal to the velocity of the particle again and continue the process. Since $v_x \ne 0$ and $|v|$ are kept constant under subsequent iterations of the third situation and upon collisions with straight walls, the third situation can occur at most finite number of times; then either the particle will either exit the system or will collide with a disk non-tangentially. In the later case, by Lemma \ref{lemma:proper projected particle path}, there exists a proper projected particle path from the non-tangential collision point to an exit. Denote the path traced $\Gamma$ by the $j^{th}$ particle in this construction by $\gamma_j$.

Again, the times of all collisions of $j^{th}$ particle along $\gamma_j$ are fixed; let $\tau^j_1, \cdots, \tau^j_{n(j)}$ be the times of collisions with disks, $D_{k(j,1)}, \cdots, D_{k(j,n(j))}$ and let $\omega^j_1, \cdots, \omega^j_{n(j)}$ be the required angular velocities (with $R \omega^j_i = v=v_t$ at tangential collisions). In the following Lemma, we again assume that at time $0$, the $j^{th}$ particle is the only particle in the system in order to ensure that the system is defined at all times.

\begin{lemma} \label{lemma: open set of paths tangential}
Assume the $j^{th}$ particle has $m(j) \geq 0$ tangential collisions before it has a non-tangential collision or exits the system. Then there exist open neighborhoods $I^j_1$ of $\omega^j_1$, $\cdots$, $I^j_{n(j)}$ of $\omega^j_{n(j)}$ such that for any choice of angular velocities $(\omega^j_1)' \in I^j_1$, $\cdots$, $(\omega^j_{n(j)})' \in I^j_{n(j)}$, if a particle starts with its initial position and velocity $(q^0_j,v^0_j)$, it possibly collides with $D_{k(j,1)}, \cdots, D_{k(j,m(j))}$ set to $(\omega^j_1)'$, $\cdots$, $(\omega^j_{m(j)})'$, collides with $D^j_{k(j,m(j)+1)}, \cdots, D^j_{k(j,n(j))}$ set to $(\omega^j_{m(j)+1})'$, $\cdots$, $(\omega^j_{n(j)})'$, and exits the system.
\end{lemma}

\begin{proof}
By Lemma \ref{lemma: open set of paths} there exist open neighborhoods $I^j_{m(j)+1,0}$ of $\omega^j_{m(j)+1}$, $\cdots$, $I^j_{n(j),0}$ of $\omega^j_{n(j)}$ such that for any choice of angular velocities $(\omega^j_{m(j)+1})' \in I^j_{m(j)+1}$, $\cdots$, $(\omega^j_{n(j)})' \in I^j_{n(j)}$ if a particle starts with $(q^0_j,v^0_j)$, it possibly collides with $D_{k(j,1)}$, $\cdots$, $D_{k(j,m(j))}$ set to $\omega^j_1, \cdots, \omega^j_{m(j)}$, collides with $D_{k(j,m(j)+1)}$, $\cdots$, $D_{k(j,n(j))}$ set to $(\omega^j_{m(j)+1})' \in I^j_{m(j)+1,0}$, $\cdots$, $(\omega^j_{n(j)})' \in I^j_{n(j),0}$, and exits the system.

At tangential collisions, we chose to set $R$ times the angular velocities of the disks to be equal to the velocities of the colliding particles, $R \omega=v=v_t$. If we vary the angular velocities of the disks at tangential collisions keeping sign the same, the particle still follows the same path and exit the system. \qed

If a particle starts at $(q^j_0,v^j_0)$ and the collision with disk $D_{k(j,1)}$ is tangential, letting the angular velocity of $D_{k(j,1)}$ to vary through $I^j_1$, we get an open set of possible collision times with $D_{k(j,2)}$ since the particle follows the same path, only with different speeds. The remaining tangential collisions are treated analogously and non-tangential collisions as in subsection \ref{subsect: flush particles out: no tang collisions}.

Therefore we can always pick a collection of $(\omega^j_i)' \in I^j_i$ such that no simultaneous collisions with same disks occur for all particles in the system. Given such a choice, each particle in $X$ follows a path from its initial position to an exit. Setting the angular velocities of the disks to $(\omega^j_i)'$ at appropriate times using the Lemma \ref{lemma: controllability}, we construct a sample path $\sigma_X$ from state $X$ to some state $X_0$ in $\Omega_0$.This completes the proof of Proposition \ref{prop:sample path flush particles out}. \qed
\end{proof}

\subsection{Proof of the Lemma \ref{lemma: controllability} ($\eta=1$)} \label{subsect:main lemma proof}

The proof is by induction on $j$, $1 \leq j \leq N$.

\textbf{j=1:} There are many ways to treat this case, but we choose a very specific one that will be a useful step in treating the induction step. Our method here is identical to the one in proof of Lemma 4.3 in \cite{Eckmann}.

Assume first $w'=0$, i.e. we want to hit $D_1$ radially. Send a particle parallel to the $x$-axis to hit the disk $D_1$ at $(1,0)$. If the initial velocity $v=(v_x,0)$ is big enough compared to $R \omega$ and the distance from $\gamma_L$ to $D_1$, the particle will be able to exit the playground without hitting $\partial \Gamma \setminus (\gamma_L \cup \gamma_R)$ again. Note that the larger $v_x$ is, the smaller is the angle of reflection; we can introduce any bound on the angle by choosing $v_x$ large enough. When $w' \not = 0$, we send a particle such that it also hits the disk at $(1,0)$. We can introduce any bound on the angle of incidence in a similar way. Clearly, by making $v_x$ sufficiently large and large enough compared to $R w'$, we can complete the above procedure in an arbitrarily short time.

\begin{center}
\begin{figure}
  \begin{center}
  \includegraphics[width=4in]{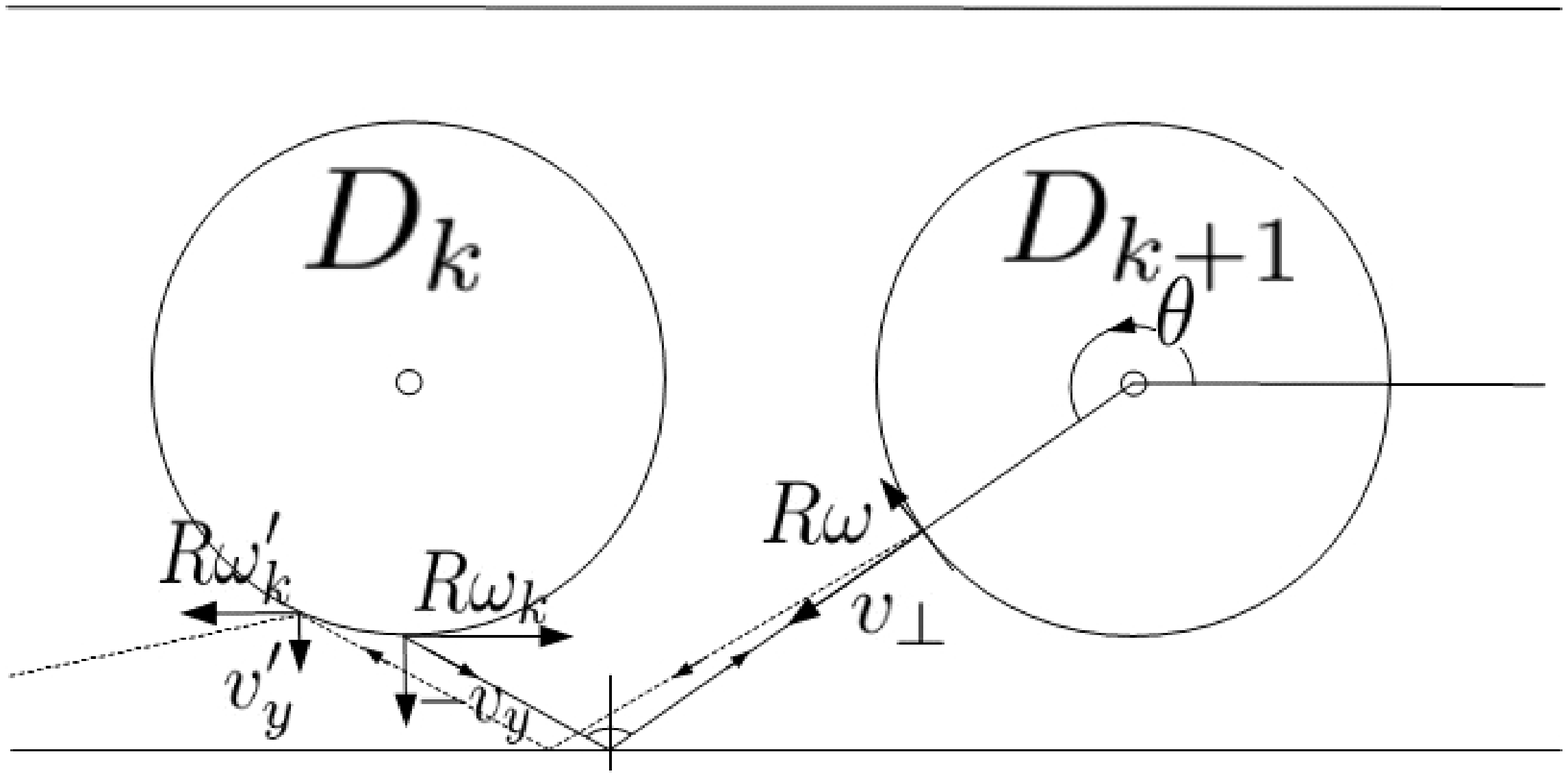}\\
  \caption{$\omega'=0$ case: the incoming trajectory $L$ is marked by a solid line; outgoing $L'$ - by a dashed line; velocities are labeled after collisions}\label{fig:DkDkplus1}
  \end{center}
\end{figure}
\end{center}

\textbf{Induction step:} Assume the lemma holds for all $j \leq k$. We would like to show that it also holds for $j=k+1$. We want to send a particle with velocity $v=(v_x,v_y)$ such that it first hits disk $D_k$ at $(2k-1,-R)$ without hitting $\partial \Gamma$ along the way. The velocity after collision is $(R \omega_k,-v_y)$, where $\omega_k$ is the angular velocity of disk $D_k$, which we can set to any value in arbitrarily short time by the induction assumption.

Consider first the case $\omega'=0$ illustrated in Figure \ref{fig:DkDkplus1}. Then we want to hit disk $D_{k+1}$ radially by following the unique trajectory that reflects from the lower boundary of the playground $\Gamma$, $[0,2N] \times \{-1\}$, once before hitting $D_{k+1}$; call this trajectory $L$. If $\theta$ is the angular position of the collision counting counterclockwise from the $x$-axis and $v_\perp$ denotes the normal component of the velocity pointing outwards, then:
$$R \omega'=-v_y \cos(\theta) + R \omega_k \sin(\theta)=0 \; \; \; \Rightarrow \; \; \; \tan(\theta)=\frac{v_y}{R \omega_k}$$
Also
$$v_\perp=-v_y \sin(\theta) - R \omega_k \cos(\theta)=-\tan(\theta) R \omega_k \sin(\theta) - R \omega_k \cos(\theta)=-\frac{R \omega_k}{\cos(\theta)}$$
implying that
$$v_x'=v_\perp \cos(\theta) + R \omega \sin(\theta)=-R \omega_k + R \omega \sin(\theta),$$
$$v_y'=v_\perp \sin(\theta) - R \omega \cos(\theta)=-v_y -R \omega \cos(\theta),$$
where $v_x'$ and $v_y'$ are the components of the velocity after the collision with disk $D_{k+1}$.

If we choose $v_y$ and $R \omega_k$ large compared to $R \omega$, the particle will follow a trajectory $L'$ very close to the reversed $L$ on the way from disk $D^{k+1}$ to disk $D_k$. By bounding the angle of reflection from disk $D_{k+1}$ (choosing $v_y$ and $R \omega_k$ as large as we need), we can ensure that $L'$ hits disk $D_k$ is a small neighborhood of $(2k-1,-R)$. During the flight of the particle to and from the disk $D_{k+1}$, we can reset the angular velocity of disk $D_k$ to a new value $R \omega_k'$ with $|R \omega_k'|$ large enough so that after the second collision with $D_k$, the particle leaves the system with no additional collisions. By choosing $v_x$, $v_y$, $R \omega_k$ and $R \omega_k'$ sufficiently large and large enough compared to $w$, we can ensure that this procedure can be done in an arbitrarily short time.

Suppose now that $\omega' \not = 0$. Then we have to hit the disk $D_{k+1}$ at a slightly different angular position $\theta'$ due to the playground geometry. Then
$$R \omega'=-v_y\cos(\theta')+R \omega_k \sin(\theta') \; \; \; \Rightarrow \; \; \; v_y=\frac{R \omega_k \sin(\theta')-R \omega'}{\cos(\theta')}$$
$$v_\perp=-v_y \sin(\theta') - R \omega_k \cos(\theta')=$$
$$=-(\frac{R \omega_k \sin(\theta')-R \omega'}{\cos(\theta')})\sin(\theta') - R \omega_k \cos(\theta')=-\frac{R \omega_k-R \omega'\sin(\theta')}{\cos(\theta')}$$
Implying that
$$v_x'=v_\perp \cos(\theta') + R \omega \sin(\theta')=-R \omega_k+(R \omega'+R \omega) \sin(\theta')$$
$$v_y'=v_\perp \sin(\theta') - R \omega \cos(\theta')=-v_y+(R \omega'-R \omega) \cos(\theta')$$
If we choose $\theta'$ sufficiently close to $\theta$, $v_y$ and $R \omega_k$ large compared to $R \omega$ and $R \omega'$, we can ensure that trajectories from and to disk $D_k$ are very close to the to $L$. Therefore, on the way back, the particle hits disk $D_k$ in a small neighborhood of $(2k-1,-R)$ and, by choosing $|R \omega_k'|$ large enough, we can send the particle out of the system without additional collisions. Again, by choosing $v_x$, $v_y$, $R \omega_k$, and $|R \omega_k|$ sufficiently large and large enough compared to $R \omega$ and $R \omega'$, we can do this procedure in an arbitrarily short time. \qed

\section{Proof of Proposition \ref{prop:flush particles out}} \label{sect: proof flush particles out}

In this section we complete the proof of Proposition \ref{prop:flush particles out}.
In section \ref{sect: flush particles out} we presented a construction of a sample path $\sigma_X$ from any admissible state $X$ to an arbitrary particle-less state $X_0 \in \Omega_0$. In subsection \ref{subsect:from any state to any state in Omega_0} we would extend this sample to a given state $Y_0 \in \Omega_0$, obtaining a sample path $\sigma$ from $X$ to $Y_0$ defined on some time interval $[0,T]$. To finish the proof of Proposition \ref{prop:flush particles out}, we are also required to show that there exists a canonical neighborhood $\Sigma$ of $\sigma$ such that each sample path in $\Sigma$ ends in $U_0$, a neighborhood of $Y_0$. In subsection \ref{subsect: proof prop 2 no tang collisions} we treat the situation when the particles in the initial state $X$ either collide a disk non-tangentially or exit the system. We finish the proof in subsection \ref{subsect: proof prop 2 with tang collisions} by treating the remaining situation with tangential collisions.

\subsection{From any state in $\Omega_0$ to any state in $\Omega_0$} \label{subsect:from any state to any state in Omega_0}

\begin{lemma} \label{lemma:from any state to any state in Omega_0}
Given $X_0, Y_0 \in \Omega_0$ and $T>0$, there exists a sample path $\sigma:X_0 \to Y_0$ on $[0,T]$.
\end{lemma}

\begin{proof}
Denote the angular positions and velocities of the disks in $X_0$ and $Y_0$ by  $(\varphi_1, \omega_1)$, $\cdots$, $(\varphi_N, \omega_N)$ and $(\varphi_1', \omega_1')$, $\cdots$, $(\varphi_N', \omega_N')$ respectively.

Divide time interval $[0,T]$ into $N$ equal subintervals of length $\frac{T}{N}$. Suppose we can set the angular position and velocity of disk $D_N$ to $(\omega_N', \tilde{\varphi}_N')$ on time interval $(0,\frac{T}{N})$, where $\varphi_N'=\tilde{\varphi}_N'+\frac{(N-1)T}{N}\omega_N'$. Then if we guarantee that no collisions occur with disk $D_N$ on time interval $[\frac{T}{N},T]$, $D_N$ will have angular position and velocity $(\varphi_N', \omega_N')$ at time $T$. Similarly we can proceed with setting angular position and velocity of disk $D_{N-1}$ to $(\omega_{N-1}', \tilde{\varphi}_{N-1}')$ on time interval $(\frac{T}{N}, \frac{2T}{N})$, where $\varphi_{N-1}'=\tilde{\varphi}_{N-1}'+\frac{(N-2)T}{N}\omega_{N-1}'$, ensuring that if no collisions happen with disk $D_{N-1}$ on time interval $[\frac{2T}{N},T]$, its angular position and velocity would be $(\varphi_{N-1}', \omega_{N-1}')$ at time $T$. And so on. In order for this procedure to work, we need the following lemma:
\end{proof}

\begin{lemma} \label{lemma: control position and velocity of one disk}
Suppose disk $D_j$ rotates with angular velocity $\omega$ and there are no particles present in the system. Given time $t>0$, angular position $\varphi'$, and angular velocity $\omega'$, there exists a sequence of particle injections on time interval $(0,t)$ from the left bath such that:
\begin{itemize}
\item at time $t$ the disk $D_j$ has the angular position and velocity $(\varphi',\omega')$,
\item at time $t$ all the injected particles have left the system,
\item on time interval $(0,t)$ the injected particles only hit disks $D_1, \cdots, D_j$.
\end{itemize}
\end{lemma}

\begin{proof}
The result of Lemma \ref{lemma: control position and velocity of one disk} is achieved by the application of the Lemma \ref{lemma: controllability} twice:

Fix some $\omega_1>\frac{3+t \omega'}{t}$. Apply the Lemma \ref{lemma: controllability} to set the angular velocity of $D_j$ to $w_1$ in time $\frac{t}{3}$. Let $\tau_1 < \frac{t}{3}$ be the time of the unique collision with disk $D_j$. Suppose we wait for some time $\tau < \frac{t}{3}$ (to be defined later) after $\frac{t}{3}$ and then apply the Lemma \ref{lemma: controllability} again to set the angular velocity of $D_j$ to $\omega'$ in time $\frac{t}{3}$. Let $\tau_2 < \frac{t}{3}$ be the time of the unique collision with disk $D_j$ counted from the time $\frac{t}{3}+\tau$.

Then at time $t$, the angular position and velocity of $D_j$ are
$$([\varphi+\tau_1 \omega + (\frac{t}{3}-\tau_1+\tau+\tau_2) \omega_1+ (\frac{2t}{3}-\tau-\tau_2)\omega'] \mod 1,\omega')$$
Let $\tilde{\varphi}=[\varphi+\tau_1 \omega + (\frac{t}{3}-\tau_1+\tau_2) \omega_1+ (\frac{2t}{3}-\tau_2)\omega'] \mod 1$; this is a fixed number since all the variables in the expression are fixed. Then we want to pick $\tau < \frac{t}{3}$ such that $\varphi'=(\tilde{\varphi}+\tau(\omega_1-\omega)) \mod 1$. This is not a problem since $\frac{1}{\omega_1 - \omega'} < \frac{t}{3}$ by the choice of $w_1$, i.e. in time $\frac{t}{3}$, a disk rotating with angular velocity $\omega_1-\omega'$ makes full revolution and thus, starting at $\tilde{\varphi}$, passes through the angular position $\varphi'$ at some time $\tau<\frac{t}{3}$. \qed

This completes the proof of Lemma \ref{lemma:from any state to any state in Omega_0}. \qed
\end{proof}

\subsection{Proof of Proposition \ref{prop:flush particles out}: No Tangential Collisions} \label{subsect: proof prop 2 no tang collisions}

In the situation when all particles in $X$ either collide with a disk non-tangentially or exit the system under the evolution of the system with no particle injections, Proposition \ref{prop:flush particles out} follows from the lemma below:

\begin{lemma} \label{lemma:continuity around proper projected particle path}
Let $\sigma$ be a sample path from a state $X$ to a state $Y$ on time interval $[0,T]$ such that each particle present in the system at any time subinterval of $[0,T]$ follows a proper projected particle path. Then for any neighborhood $U$ of $Y$, there exists a canonical neighborhood $\Sigma$ of $\sigma$ such that each sample path in $\Sigma$ ends in $U$.
\end{lemma}

\begin{theopargself}
\begin{proof}[of Lemma \ref{lemma:continuity around proper projected particle path}]
Denote by $c$ the sequence of injections that generates $\sigma$. Let $\Sigma'$ be any canonical neighborhood of $\sigma$; denote by $U_X'$ the neighborhood of $X$ and by $C'$ the canonical neighborhood of $c$ such that each sample path in $\Sigma'$ starts with an initial condition in $U_X'$ and follows a sequence of injections from $C'$. Define $f:U_X' \times C' \to \Omega$ as follows: if $\sigma' \in \Sigma'$ is a sample path that starts at state $X' \in U_X'$, is generated by a sequence of injections $c' \in C'$, and ends at state $Y'$, then $f(X',\epsilon')=Y'$. To prove lemma \ref{lemma:continuity around proper projected particle path}, it is enough to show that $f$ is continuous at $(X,c)$.

We assumed that along sample path $\sigma$ each particle follows a proper projected particle path, i.e. it is only allowed to collide with disks non-tangentially. The continuity of $f$ follows from the following facts:
\begin{itemize}
\item If a particle does not collide with $\partial \Gamma$ on time interval $[\tau_1, \tau_2]$, its position and velocity change continuously.
\item If a particle collides with the wall on time interval $[\tau_1,\tau_2]$ and it is not involved in any other collisions with $\partial \Gamma$ on time interval $[\tau_1,\tau_2]$, then its final position and velocity depend continuously of its initial position and velocity. [This fact follows from the continuity of the billiard flow at collisions]
\item If a particle collides with a disk non-tangentially on $[\tau_1, \tau_2]$ and neither the particle nor the disk is involved in other collisions on time interval $[\tau_1,\tau_2]$, then the particle's position and velocity as well as the disk's position and angular velocity depend continuously on their initial positions and velocities. [Follows from Lemma \ref{lemma:cts dependence of particle and disk's position and velocity}]
\item If a particle exits through $\gamma_L$ or $\gamma_R$ on time interval $[\tau_1,\tau_2]$ and does not collide with $\partial \Gamma$ on time interval $[\tau_1,\tau_2]$, then the coordinates of the other particles and disks are independent from the coordinates of the exiting particle on time interval $[\tau_1,\infty)$.
\item The position and velocity of an injected particle depend continuously on the injected parameters. [Follows from Lemma \ref{lemma: particle enters 4D}]
\end{itemize}

\begin{flushright}
\qed
\end{flushright}
\end{proof}
\end{theopargself}

\subsection{Proof of Proposition \ref{prop:flush particles out}: Tangential Collisions} \label{subsect: proof prop 2 with tang collisions}

When we constructed a sample path $\sigma_X$ from $X$ to some particle-less state $X_0 \in \Omega_0$ in section \ref{sect: flush particles out}, we first assigned a path in $\Gamma$ to each particle in $X$ from its initial position to an exit. In order for a particle to follow such a path, the disks had to be set to unique angular velocities at collisions (with $R \omega=v=v_t$ at tangential collisions). Then we showed that setting the disks to nearby angular velocities at collisions makes particles follow nearby paths with collisions happening at nearby times; this was crucial for choosing appropriate particle paths such that no simultaneous collisions with same disks occur.

When choosing nearby paths in order to avoid simultaneous collisions with same disks,  we might have to require that at some tangential collisions, $R$ times the angular velocities of the disks are \emph{not} equal to the velocities of the colliding particles. And near a tangential collision with $R \omega \ne v=v_t$, final position and velocity of a particle does not depend continuously on initial position and velocity unlike in the situation with $R \omega=v=v_t$. That prevents us from direct extension of Lemma \ref{lemma:continuity around proper projected particle path} to the situation when tangential collisions might occur.

However, Proposition \ref{prop:flush particles out} only requires us to choose a sample path $\sigma$ and a canonical neighborhood $\Sigma$ of $\sigma$ such that each sample path in $\Sigma$ ends in the given neighborhood $U_0$. By making the size of the discontinuity small enough, we will still be able to ensure that every sample path in $\Sigma$ ends in $U_0$.

The following two Lemmas imply Proposition \ref{prop:flush particles out}:

\begin{lemma} \label{lemma:canonical neighborhood Sigma_X}
Given $\epsilon >0$, there exists time $T_X>0$, a sample path $\sigma_X$ on $[0,T_X]$ from $X$ to some particle-less state $X_0$ and a canonical neighborhood $\Sigma_X$ of $\sigma_X$, such that each sample path in $\Sigma_X$ ends in an $\epsilon$-neighborhood $U_\epsilon(X_0)$ of $X_0$.
\end{lemma}

\begin{remark}
Note that state $X_0$ depends on the choice of $\sigma_X$, while the size $\epsilon$ of the neighborhood $U_\epsilon(X_0)$ around $X_0$ does not.
\end{remark}

\begin{lemma} \label{lemma:canonical neighborhood Sigma_0}
Given $Y_0 \in \Omega_0$ and a neighborhood $U_0$ of $Y_0$, there exists $\epsilon>0$ such that for any state $X_0 \in \Omega_0$, there exists time $T_0>0$, a sample path $\sigma_0:X_0 \to Y_0$ on $[0,T_0]$, and a canonical neighborhood $\Sigma_0$ of $\sigma_0$ in which each sample path starts in $U_\epsilon(X_0)$ and ends in $U_0$ and for any point $Y \in U_\epsilon(X_0)$, there exists a sample path in $\Sigma_0$ that starts at $Y$.
\end{lemma}

Lemma \ref{lemma:canonical neighborhood Sigma_0} follows directly from Lemmas \ref{lemma:from any state to any state in Omega_0} and \ref{lemma:continuity around proper projected particle path}.

\begin{theopargself}
\begin{proof}[of Lemma \ref{lemma:canonical neighborhood Sigma_X}]

As in section \ref{sect: flush particles out}, denote the initially assigned path in $\Gamma$ traced by the $j^{th}$ particle by $\gamma_j$. Let $\tau^j_1, \cdots, \tau^j_{n(j)}$ be the times of collisions with disks, $D_{k(j,1)}, \cdots, D_{k(j,n(j))}$ and let $\omega^j_1, \cdots, \omega^j_{n(j)}$ be the required angular velocities (with $R \omega^j_i = v=v_t$ at tangential collisions).

In the following Lemma, we assume that at time $0$ the $j^{th}$ particle is the only particle in the system in order to ensure that the system is defined at all times.

\begin{lemma} \label{lemma: open set of paths neighborhood tangential}
Assume the $j^{th}$ particle has $m(j) \geq 0$ tangential collisions before it has a non-tangential collision or exits the system. Then there exist open neighborhoods $V^j$ of $(q_j^0,v_j^0)$, $I^j_1$ of $\omega^j_1$, $\cdots$, $I^j_{n(j)}$ of $\omega^j_{n(j)}$ such that for any choice of angular velocities $(\omega^j_1)' \in I^j_1$, $\cdots$, $(\omega^j_{n(j)})' \in I^j_{n(j)}$, if a particle starts with a position and a velocity from $V^j$, it possibly collides with $D_{k(j,1)}, \cdots, D_{k(j,m(j))}$ set to $(\omega^j_1)'$, $\cdots$, $(\omega^j_{m(j)})'$, collides with $D^j_{k(j,m(j)+1)}, \cdots, D^j_{k(j,n(j))}$ set to $(\omega^j_{m(j)+1})'$, $\cdots$, $(\omega^j_{n(j)})'$, and exits the system.
\end{lemma}

We will prove Lemma \ref{lemma: open set of paths neighborhood tangential} after finishing the proof of Lemma \ref{lemma:canonical neighborhood Sigma_X}.

Let $T^j_X$ be an upper bound on the time it takes for the $j^{th}$ particle to exit $\Gamma$ along all possible paths described in Lemma \ref{lemma: open set of paths neighborhood tangential}; and let $T_X=\max_{1 \leq j \leq k}\{T^j_X\}$.
Define $g_j:V^j \times I^j_1 \times \cdots \times I^j_{n(j)} \to \Omega_0$ as follows: if a particle starts with a position and a velocity from $V^j$, and $D_{k(j,1)}, \cdots, D_{k(j,n(j))}$ are set to $(\omega^j_1)' \in I^j_1$, $\cdots$, $(\omega^j_{n(j)})' \in I^j_{n(j)}$ at potential collision times, let $g_j(q,v,(\omega^j_1)', \cdots, (\omega^j_{n(j)})')$ be the state of the system at time $T_X$. $g^j$ is continuous at $(q^j_0,v^j_0,\omega^j_1, \cdots, \omega^j_{n(j)})$; so there exist sub-neighborhoods $V^j_\epsilon$ of $V^j$, $I^j_{1,\epsilon}$ of $I^j_1$ $\cdots$, $I^j_{n(j),\epsilon}$ of $I^j_{n(j)}$ such that for any $(q_j,v_j) \in V^j_\epsilon$, $(\omega^j_1)' \in I^j_\epsilon$, $\cdots$, $(\omega^j_{n(j)})' \in I^j_{1,\epsilon}$, $$|g_j(q_j,v_j,(\omega^j_1)', \cdots, (\omega^j_{n(j)})')-g_j(q_j^0,v_j^0,\omega^j_1, \cdots, \omega^j_{n(j)})| < \epsilon/2.$$

Now we are ready to deal with $k$-particle system. As in section \ref{sect: flush particles out}, we can choose a path for each particle in $X$ such that upon each disk collision, $(\omega^j_i)' \in I^j_{i,\epsilon}$ and no simultaneous collisions with same disks occur. That defines $\sigma_X$ on $[0,T_X]$; let $X_0$ be the state where $\sigma_X$ ends. To define $\Sigma_X$ choose further sub-neighborhoods of $I^j_{i,\epsilon}$'s to ensure that each sample path is $\Sigma_X$ is defined up to time $T_X$. Then each sample path in $\Sigma_X$ ends in an $\epsilon$-neighborhood $U_\epsilon(X_0)$ of $X_0$ by the above inequality. \qed
\end{proof}

\begin{proof}[of Lemma \ref{lemma: open set of paths neighborhood tangential}]

By Lemma \ref{lemma: open set of paths} and the fact the near a tangential collision with $R \omega=v=v_t$, particle's final position and velocity depend continuously on its initial position and velocity, there exist open neighborhoods $V^j_0$ of $(q^0_j,v^0_j)$, $I^j_{m(j)+1,0}$ of $\omega^j_{m(j)+1}$, $\cdots$, $I^j_{n(j),0}$ of $\omega^j_{n(j)}$ such that for any choice of angular velocities $(\omega^j_{m(j)+1})' \in I^j_{m(j)+1}$, $\cdots$, $(\omega^j_{n(j)})' \in I^j_{n(j)}$ if a particle starts with a position and a velocity from $V^j_0$, it possibly collides with $D_{k(j,1)}, \cdots, D_{k(j,m(j))}$ set to $\omega^j_1, \cdots, \omega^j_{m(j)}$, collides with $D_{k(j,m(j)+1)}$, $\cdots$, $D_{k(j,n(j))}$ set to $(\omega^j_{m(j)+1})'$, $\cdots$, $(\omega^j_{n(j)})'$, and exits the system.

Now we would like to allow an open neighborhood of angular velocities around the $m^{th}$ tangential collision. The neighborhood $V^j_0$ can be split into two parts: $(V^j_0)_c \sqcup (V^j_0)_{nc} = V^j_0$, where $(V^j_0)_c$ denotes the set of initial positions and velocities such that a particle with a position and a velocity from $(V^j_0)_c$ will collide with $D^j_{k(j,m(j))}$ provided that $D^j_{k(j,1)}, \cdots, D^j_{k(j,m(j)-1)}$ are set to angular velocities $\omega^j_1, \cdots, \omega^j_{m(j)-1}$ before potential collisions.

Then the position and velocity of a particle after collision $D_{k(j,m(j))}$ depend continuously on its initial position and velocity in $(V^j_0)_c$ even if the angular velocity of $D_{k(j,m(j))}$ is not equal to $\omega^j_{m(j)}$. Also, since particles in $(V^j_0)_{nc}$ do not collide with $D_{k(j,m(j))}$, they exit the system provided the angular velocities of the disks $D_{k(j,1)}$, $\cdots$, $D_{k(j,m(j)-1)}$ are set to $\omega^j_1$, $\cdots$, $\omega^j_{m(j)-1}$ at appropriate times and angular velocities of the disks $D_{k(j,m(j)+1)}$, $\cdots$, $D_{k(j,n(j))}$ are set to values from $I^j_{m(j)+1,0}$, $\cdots$, $I^j_{n(j),0}$ before collisions.

Therefore there exists an open neighborhood $I^j_{m(j)}$ of $\omega^j_{m(j)}$ and open subneighborhoods $V^j_{m(j)}$ of $V^j_0$, $I^j_{m(j)+1,m(j)}$ of $I^j_{m(j)+1}$, $\cdots$, $I^j_{n(j),m(j)}$ of $I^j_{n(j)}$ such that if a particle starts with a position and a velocity from $V^j_{m(j)}$, it possibly collides with $D^j_{k(j,1)}$, $\cdots$, $D^j_{k(j,(m(j)-1))}$ with angular velocities $\omega^j_1$, $\cdots$, $\omega^j_{m(j)-1}$, possibly collides with $D^j_{k(j,m(j))}$ with angular velocity $(\omega^j_{m(j)})' \in I^j_{m(j)}$, collides with $D_{k(j,m(j)+1)}$, $\cdots$, $D_{k(j,n(j))}$ with angular velocities $(\omega^j_{m(j)+1})' \in I^j_{m(j)+1,m(j)}$,   $\cdots$, $(\omega^j_{n(j)})' \in I^j_{n(j),m(j)}$, and exits the system. The remaining tangential collisions are treated similarly. \qed
\end{proof}
\end{theopargself}

\begin{acknowledgement}
I would like to thank my Ph.D. thesis advisor Lai-Sang Young for proposing the problem, fruitful discussions, effective criticism, and useful comments on many drafts of this paper. This work was partially supported by the National Science Postdoctoral Research Fellowship.
\end{acknowledgement}

\end{document}